\documentclass[preprint,12pt,authoryear]{elsarticle}

\usepackage{amssymb}
\usepackage{amsthm}
\usepackage{amsmath}
\newtheorem{theorem}{Theorem}[section]

\newtheorem{proposition}[theorem]{Proposition}
\newtheorem*{remark}{Remark}
\theoremstyle{definition}
\newtheorem{definition}{Definition}[section]

\usepackage{thmtools}

\DeclareMathOperator*{\argmax}{arg\,max}
\DeclareMathOperator*{\argmin}{arg\,min}

\usepackage{algorithm}
\usepackage{algpseudocode}

\usepackage{hyperref}
\usepackage{svg}

\usepackage{float}
\usepackage{placeins}

\usepackage{lineno}
\usepackage{censor}
\ifdoubleblind
    \newcommand{\identifier}{\censor}
\else
    \newcommand{\identifier}{}
\fi

\usepackage{xurl}

\journal{Games and Economic Behaviour}

\begin{document}

\begin{frontmatter}



\title{A Game of Competition for Risk}


\author[inst1,inst2]{Louis Abraham}
\ead{louis.abraham@yahoo.fr}

\address[inst1]{Université Paris 1 Panthéon--Sorbonne}

\address[inst2]{Association Française d'Épargne et de Retraite (AFER)}




\begin{abstract}
In this study, we present models where participants strategically select their risk levels and earn corresponding rewards, mirroring real-world competition across various sectors. 
Our analysis starts with a normal form game involving two players in a continuous action space, confirming the existence and uniqueness of a Nash equilibrium and providing an analytical solution. We then extend this analysis to multi-player scenarios, introducing a new numerical algorithm for its calculation. 
A key novelty of our work lies in using regret minimization algorithms to solve continuous games through discretization. This groundbreaking approach enables us to incorporate additional real-world factors like market frictions and risk correlations among firms. We also experimentally validate that the Nash equilibrium in our model also serves as a correlated equilibrium.
Our findings illuminate how market frictions and risk correlations affect strategic risk-taking. We also explore how policy measures can impact risk-taking and its associated rewards, with our model providing broader applicability than the Diamond-Dybvig framework. We make our methodology and code open-source\footnote{available at \identifier{\url{https://github.com/louisabraham/cfrgame}}}.
Finally, we contribute methodologically by advocating the use of algorithms in economics, shifting focus from finite games to games with continuous action sets. Our study provides a solid framework for analyzing strategic interactions in continuous action games, emphasizing the importance of market frictions, risk correlations, and policy measures in strategic risk-taking dynamics.

\end{abstract}




\end{frontmatter}

\section{Introduction}
\label{sec:introduction}

Risk-taking during competition is an everyday occurrence, spanning numerous scenarios from financial markets to environmental policies. In these settings, individuals and organizations must balance the lure of potential rewards against the potential for negative outcomes such as bankruptcy or ecological disasters. Understanding and predicting behaviors in these contexts is crucial for an array of parties, including policymakers, regulators, and investors.

Game theory provides a compelling lens for analyzing these situations. It helps model strategic interactions among players and outlines the incentives prompting their actions. Our work focuses on normal form games - situations where each player selects a strategy and earns a payoff based on the collective actions of all players. In this article, we explore continuous models of competition, where players can choose their level of risk, receiving higher rewards for taking on more risk.

Nash equilibrium is a key concept for our study. It's a state of stability in the game, where no player sees an advantage in deviating from their chosen strategy. In a normal form game, a Nash equilibrium consists of strategies where each player's strategy is the best response to the strategies of others. This concept is fundamental to game theory and has been widely used in various fields like economics and political science to model strategic behavior \citep{moulin1986game, varoufakis2008game}.

Our exploration starts with a straightforward normal form game involving just two players. For this setup, we provide solid proof for both the existence and uniqueness of a Nash equilibrium, and we go further by presenting an analytical solution. This simple model serves as our fundamental building block, a starting point that offers a solid base of understanding.

Subsequently, we enhance our model to incorporate the complexity of multiple players. This extension allows us to probe deeper into the strategic dynamics in more realistic, multi-actor competitive environments. Even with the additional complexity, we manage to maintain the uniqueness of the Nash equilibrium and solve the game analytically.

The third stage of our investigation introduces two vital real-world components: market frictions and risk correlations among firms. We begin by defining these elements in a two-player context, paving the way for more complex scenarios.

The final phase of our study marks a significant departure from conventional approaches. Given the complexities introduced by market frictions and risk correlations, we adopt a novel technique—using regret minimization algorithms to discretize and solve our game. This innovation, which opens new vistas in the study of strategic interactions, proves especially valuable in the face of the potentially intractable analytical solutions that these intricate scenarios might present. 

Our experimental validation establishes that the Nash equilibria in our model also function as correlated equilibria, endorsing the use of correlated equilibria to model strategic behavior. To compute these equilibria, we employ an array of algorithms, prominently featuring regret matching and counterfactual regret minimization, thus highlighting the expanding potential of algorithmic solutions for tackling complex strategic interactions.

Next, we examine the impact of penalties and market frictions on strategic behavior and results in our continuous model. We find that penalties reduce both the average risk taken by players and their total rewards. Market frictions, on the other hand, lower average risk but increase total rewards. These frictions have a more significant effect on total rewards in high-penalty environments. In especially inefficient markets with high market frictions, raising penalties can promote cooperation and increase total rewards.

We also assess the effects of risk correlations among firms on strategic behavior and performance. We find that players take more risks in negative correlation situations, which boosts their payoff compared to a no-correlation scenario. On the flip side, in positively correlated settings, risk-taking is reduced. The impact on performance varies, being negative in efficient markets but potentially positive in less predictable markets.

Our model interestingly aligns with the Diamond-Dybvig framework, where financial institutions can choose a parameter affecting their utility function and their likelihood of bankruptcy. This parallel allows our model to explore situations such as competition among banks over deposit contract interest rates, akin to the dynamic modeled by Diamond and Dybvig. But our model is distinct and more generalized, focusing not on specific financial metrics, but on a broader notion of failure probability, enabling us to explore strategic competition dynamics in a broader array of scenarios beyond baking.

Our findings offer valuable insights for policymakers, regulators, and investors who need to understand behavior in competitive, risk-laden situations. We highlight the significant influence of penalties and market frictions on strategic behavior and outcomes, and show how risk correlations can considerably alter strategic behavior and performance in competitive dynamics. By clarifying these elements, we contribute to the discussion on how to design effective interventions and policies that encourage cooperation and improve outcomes in competitive situations involving risk.

\section{A simple model of competition for risk}
\label{sec:simplemodel}
\subsection{Description}

In this section, we introduce a simple model of competition for risk that serves as a backbone for our study. We consider a situation where two actors, denoted as Player 1 and Player 2, engage in competition by taking actions that make them more attractive to customers but also increase their risk of failure. For example, firms may choose to lower their prices to attract more customers but in doing so, they increase the likelihood of not being able to repay their loans. Similarly, insurance companies may lower their premiums to attract more customers but this comes at the cost of a higher risk of failure. Banks may increase their deposit rates to attract more customers but this also increases their vulnerability to liquidity crises.

In our model, each player directly sets their failure probability, denoted as $r_p$. While this assumption may not be realistic in practice, we note that in many situations, firms use models that map real-world actions, such as setting prices or premiums, to failure probabilities. This mapping is often a monotonous function that can be inverted to yield real-world actions from failure probabilities, making our model practical.
Based on the failure probabilities set by each player, the players can randomly ``lose'' the game. In our simple model, this translates into being applied a penalty, denoted as $P$. We assume $P>0$. We assume that the failure events are independent, meaning that each player draws a uniform random variable $f_p$ from the interval $[0,1]$ and fails if $f_p < r_p$. We will later introduce correlations between the variables $f_p$ to model real-life situations where correlations may be positive or negative.

After the failure events are determined, the players that did not fail compare their risk levels, and the player that played the highest risk level is rewarded with a payoff, denoted as $R$. Since the game is unchanged when scaling both $P$ and $R$, we assume $R = 1$. In the case of ties between risk levels, we consider several ways of resolving them, such as none of the players receiving the reward, the reward being shared equally between them, or the reward being randomly given to one of them.

We note that, as shown later, the optimal strategies in our model are modeled by real distributions, which means that the probability of ties is zero. However, when we use discrete action sets to compute approximate Nash equilibria, the action sets can overlap, and we implement the first two variations of resolving ties (the last two are equivalent in expectation). This simple model serves as a foundation for our study, and we will extend it by introducing correlations between the players' failure probabilities and market frictions in subsequent sections.
For two players, assuming $r_1 > r_2$, the outcome matrix will be:

\begin{center}
\begin{tabular}{|c|c|c|}
\cline{2-3}
\multicolumn{1}{c|}{} & $f_1\ge r_1$ & $f_1<r_1$ \\
\hline
$f_2\ge r_2$ & $R=1, 0$ & $-P, 1$ \\
\hline
$f_2<r_2$ & $1, -P$ &  $-P, -P$ \\
\hline
\end{tabular}
\end{center}

Each cell contains the rewards to each player. For example, in the upper left cell, no failure happens. Since we assumed $r_1 > r_2$, player 1 gets $R=1$ and player 2 gets $0$.

\subsection{Equivalence to a normal-form game}

We can represent the game described above in the framework of extensive-form games \citep{hart1992games} by modeling the drawing of the random variables $f_p$ using Chance nodes. Since the outcomes are subject to randomness, it is natural to assume that the actors operate under the expected utility hypothesis, which implies that they possess a von Neumann–Morgenstern utility function \citep{neumann1944theory}. Consequently, we can define a normal-form game with payoffs equal to the expected payoffs of the corresponding extensive-form game. By doing so, we can leverage the theory of normal-form games and apply various solution concepts, such as Nash equilibria, to analyze the competition between the actors.
\begin{proposition}
The expected utilities $u_p$ are computed as follows in the 2-player game:
\begin{align*}
    u_2(r_1, r_2) &= u_1(r_2, r_1) \text{~(symmetry)}\\
    u_1(r_1, r_2) &= r_2 (1-r_1) R - r_1 P + [ r_1 > r_2 ] (1-r_1)(1-r_2) R
\end{align*}
where $[ \cdot ]$ is the Iverson bracket.
\end{proposition}
\begin{proof}
Player 1 can fail with probability $r_1$, in which case they lose $P$. If Player 2 loses and Player 1 does not, which happens with probability $r_2(1-r_1)$, Player 1 wins $R$. Finally, if none of the players fails, when $r_1 > r_2$, Player 1 can win $R$.
\end{proof}

It is possible to encompass the shared payoff in case of ties by defining the Iverson bracket to be $\frac{1}{2}$ when $r_1 = r_2$.
Figure \ref{fig:reward} shows what the reward function of Player 1 looks like when Player 2 adopts the fixed strategy $r_2=0.2$.

\begin{figure}[htbp]
    \centering
    \includesvg[width=0.8\linewidth]{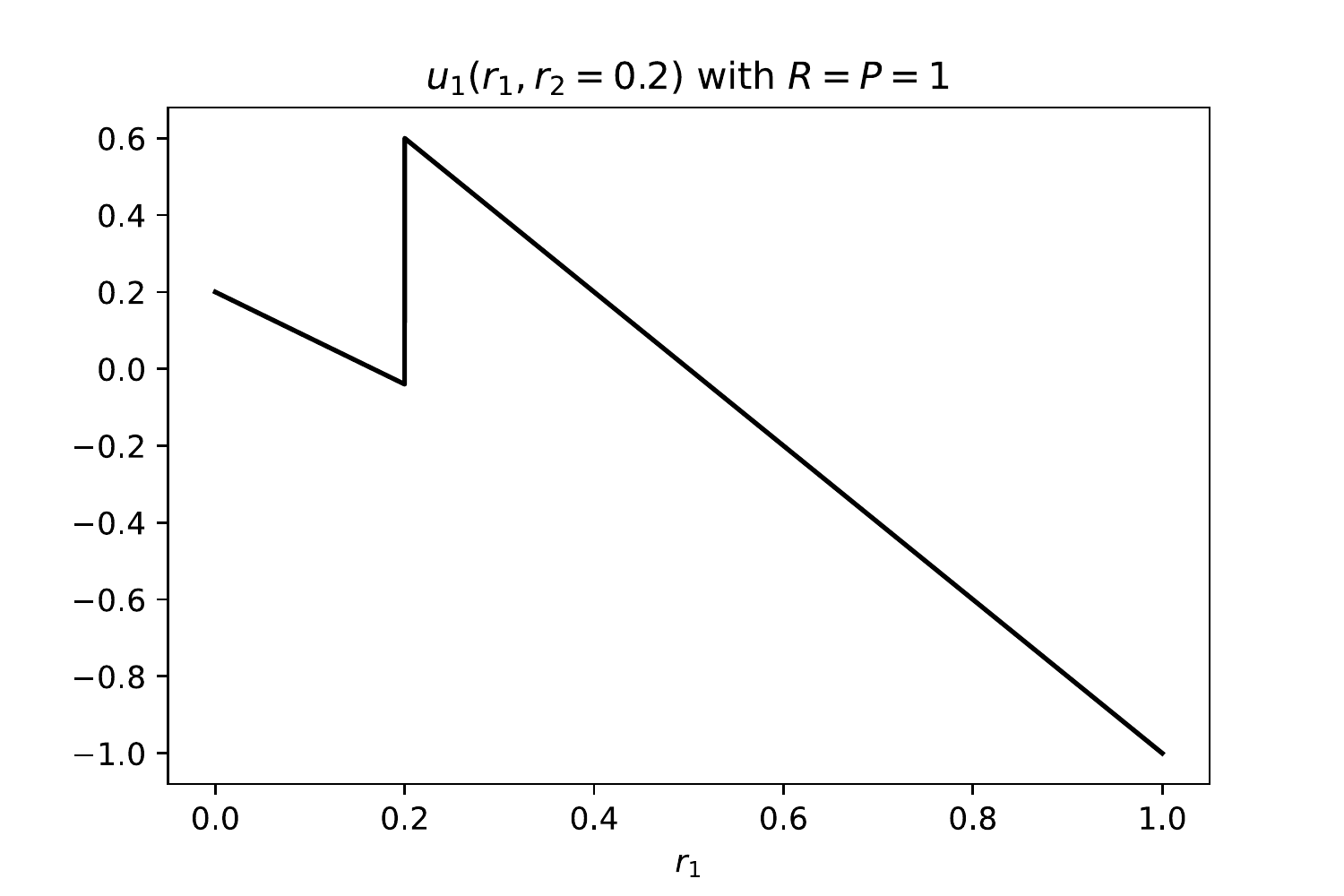}
    \caption{Reward function}
    \label{fig:reward}
\end{figure}

The discontinuity of our game is similar to two games: the War of Attrition game from \cite{smith1974theory} and the visibility game from \cite{Lotker2008-tx}. In the War of Attrition game, each player independently chooses a time to quit the game. The player who stays in the game for the longest time wins a prize. However, both players incur a cost that increases over time while they are still in the game. In the visibility game, the payoff of each player is the difference with the next player, or 1 for the player that plays the largest move. A major difference between our game and those two games is that we model the probability of failure. This means, for example, that the player taking less risk can still win the reward if the first player fails.

However, the structure of our problem and the analytical solution of the Nash equilibrium are similar to \cite{Lotker2008-tx}. We name our game the Competition for Risk game and will write it CfR in the rest of the article.

\subsection{Nash equilibrium}

A Nash equilibrium is a set of strategies, one for each player, such that no player can improve their payoff by unilaterally changing their strategy, given the strategies of the other players. In other words, each player's strategy is the best response to the strategies chosen by the other players. Nash equilibria are important because they provide a way to predict the outcome of a game if each player acts rationally and selfishly. They can also help explain why certain outcomes occur in real-world situations.

In our model of competition for risk, finding Nash equilibria can help us understand how firms, banks, and insurance companies behave when they compete for prices and take different levels of risk. By analyzing the Nash equilibria of our model, we can predict how different players will act and what the resulting outcomes will be. Moreover, we can compare the efficiency of different equilibria and use them as a benchmark to evaluate the performance of different strategies.

As in the game of \cite{Lotker2008-tx}, we can prove that there is no pure Nash equilibrium, that is, a deterministic optimal strategy.

\begin{theorem}
The CfR game does not admit any pure Nash equilibrium.
\end{theorem}
\begin{proof}
Suppose the existence of an equilibrium $s_1, s_2$. Suppose that $s_1 > s_2$. Then Player 1 can improve their payoff by playing $s_1 - \varepsilon$ since they still get the reward and take less risk. By symmetry, this implies that $s_1 = s_2$. If $s_1 < 1$, then player 1 can improve their situation by playing $s_1 + \varepsilon$ since they get $R$ (or $\frac{R}{2}$ if the reward is shared). If $s_1 = 1$ then the payoff is $-P < 0$ with probability $1$ so it is better to play $0$ which gives payoff $0$ with probability $1$.
\end{proof}

\begin{definition}A strategy $s$ (a couple of strategies) is Pareto optimal if there is no other strategy $s'$ such that $\forall p, u_p(s) \le u_p(s')$ and $\exists p, u_p(s) < u_p(s')$. It is $\varepsilon$-Pareto optimal if there is no strategy $s'$ such that $\forall p, u_p(s) \le u_p(s')$ and $\exists p, u_p(s) + \varepsilon < u_p(s')$.
\end{definition}
\begin{remark} If the reward is shared in case of tie, the pure strategy $(0,0)$ gives reward $\frac{R}{2}$ to each player. This strategy is Pareto-optimal.
\end{remark}
\begin{theorem}
\label{thm:pareto}
For every $\varepsilon$, there is a $\varepsilon$-Pareto optimal strategy that gives $\frac{R - \varepsilon}{2}$ to each player.
\end{theorem}
\begin{proof}
Let us consider the joint mixed strategy where each player plays uniformly at random in the interval $[0, 2 \varepsilon]$. The payoff is
\begin{align*} \mathbb{E}[u_1] &= \mathbb{E}\left[ r_2 (1-r_1) R - r_1 P + [ r_1 > r_2 ](1-r_1)(1-r_2) R\right] \\ &= \varepsilon (1-\varepsilon) R - \varepsilon P + \frac{(1-\varepsilon)^2}{2} R \\
&\rightarrow_{\varepsilon \rightarrow 0} \frac{R}{2} \end{align*}
so by taking $\varepsilon$ small enough we can get as close to $\frac{R}{2}$ as we want.

If each player gets payoff $\frac{R - \varepsilon}{2}$then no player can get $\varepsilon$ without degrading the other’s performance else the total payoff would be more than $R$.
\end{proof}

However, the $\varepsilon$-Pareto strategy is highly concentrated around $0$, incentivizing players to deviate and increase their chances of winning $R$ without taking on additional risk. Thus, this strategy fails to form a Nash equilibrium.

Fortunately, the CfR game possesses a unique Nash equilibrium, a powerful property that showcases the strength of our approach. Moreover, this equilibrium is symmetric.

For finite games, \citet{Nash1950-jp} proved the existence of mixed Nash equilibria, while Glicksberg's theorem \citep{Glicksberg1951-wp} extended  this result to continuous reward functions. \citet{Dasgupta1986-gu} established conditions under which discontinuous games can possess Nash equilibria and symmetric games can admit symmetric equilibria.

The uniqueness of the Nash equilibrium is a highly desirable property, with most models using concave reward functions to ensure it. Therefore, it is noteworthy that the CfR game exhibits a unique Nash equilibrium.

We recall Theorem 2.1 from \citet{Lotker2008-tx}:

\begin{theorem}
\label{thm:jules}
Let $(f_1, \ldots , f_n)$ be a Nash equilibrium point, with expected payoff $u_i^*$ to Player $i$ at the equilibrium point. Let $u_i(x)$ (as an abuse of notation) denote the expected payoff for Player $i$ when he plays the pure strategy $x$ and all other players play their equilibrium mixed strategy. Then $u_i(x) \leq u_i^*$ for all $x \in [0, 1]$, and furthermore, there exists a set $\mathcal{Z}$ of measure $0$ such that $u_i(x) = u_i^*$ for all $x \in support(f_i) \setminus \mathcal{Z}$.
\end{theorem}

This theorem means that at the Nash equilibrium, almost any move that is in the support of a player's strategy should give them the same (maximal) payoff. This theorem is crucial to find the equilibrium in the CfR game.

\begin{restatable}{theorem}{nashcor}
\label{thm:nashcor}
Up to a set of measure zero, the CfR game admits a unique Nash equilibrium. This equilibrium is symmetric and its distribution is $f(x) = \left[x < 1 - \sqrt{\frac{k - 1}{k + 1}}\right] \frac{ k - 1}{(1-x)^3}$ with $k := \sqrt{(P + 1)^2 + 1}$. The the average move is $\bar r = k - (P+1)$ and the utility of each player is $u^* = \bar r$.
\end{restatable}

\begin{proof}
    See \ref{proof:nashcor} for a full proof. For a less rigorous treatment, refer to the proof of the more general Theorem \ref{thm:multiple}.
\end{proof}

At $P = 1$, the cutoff value is $1 - \sqrt{\frac{\sqrt{5} - 1}{\sqrt{5} + 1}} = 2 - \phi \approx 0.382$ with $\phi$ the Golden ratio. We plot the distribution in Figure \ref{fig:nash}.

\begin{figure}[htbp]
  \centering
  \begin{minipage}[t]{0.48\textwidth}
    \centering
    \includesvg[width=1.1\linewidth]{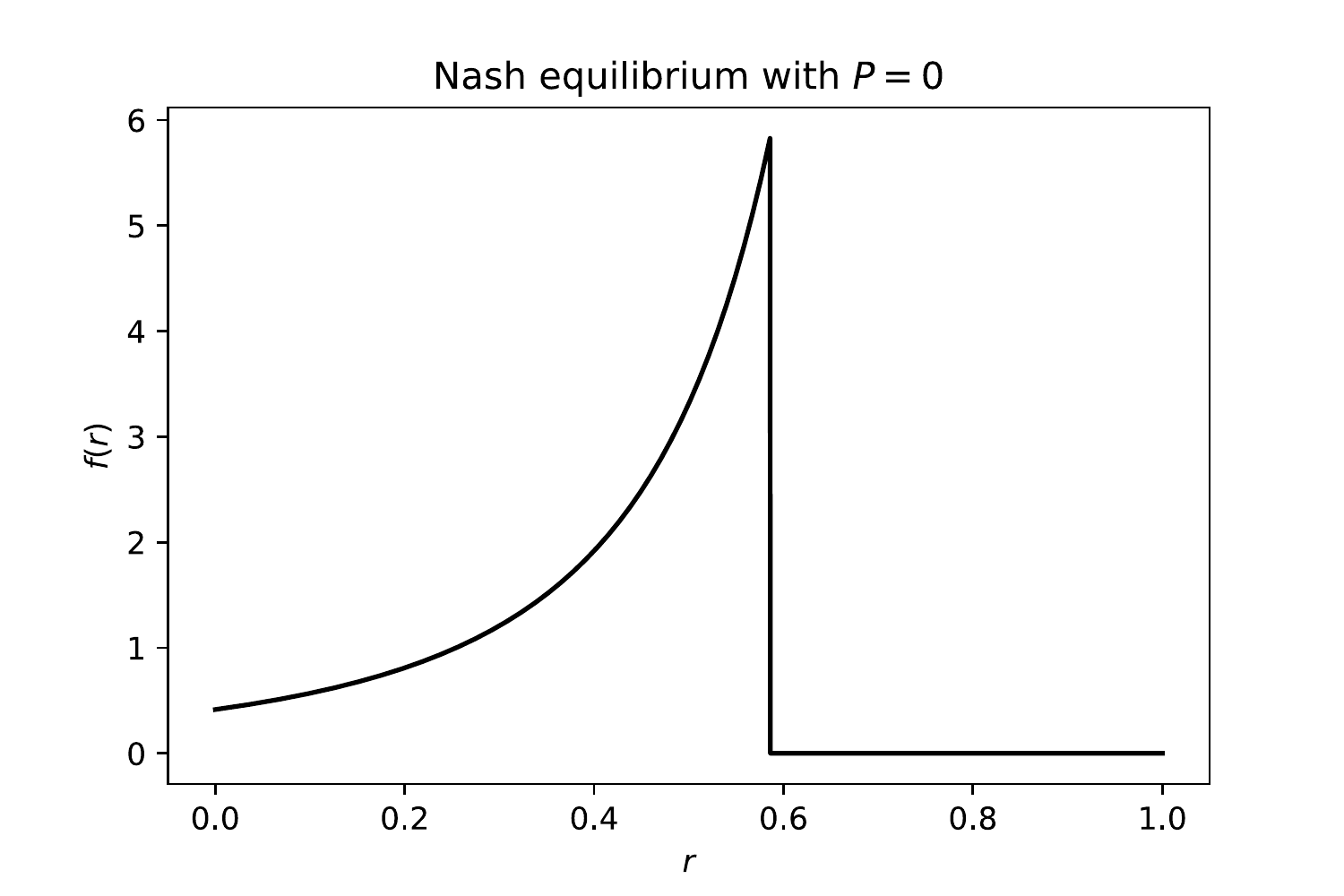}
  \end{minipage}
  \hfill
  \begin{minipage}[t]{0.48\textwidth}
    \centering
    \includesvg[width=1.1\linewidth]{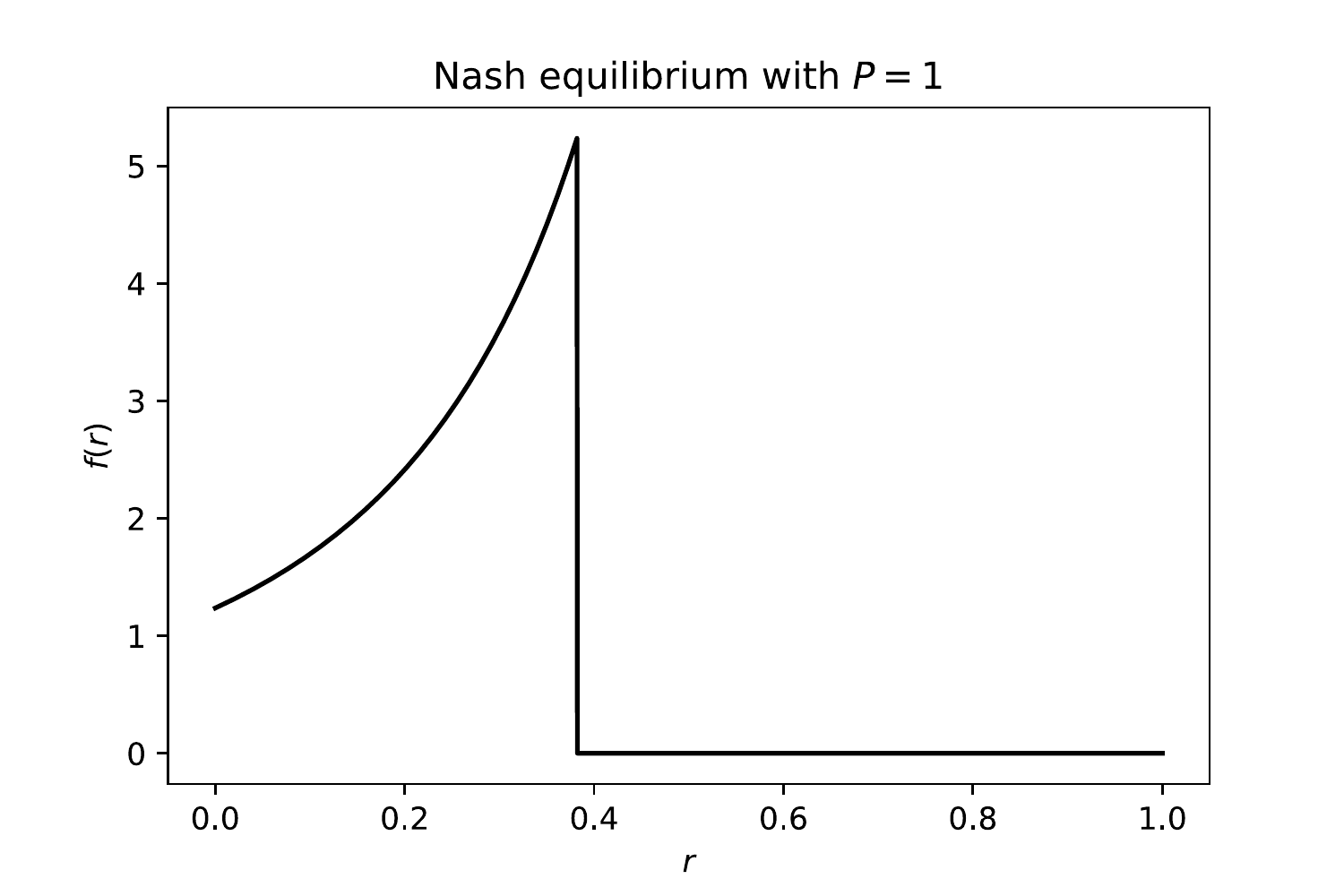}
  \end{minipage}
  \caption{Nash equilibrium}
  \label{fig:nash}
\end{figure}

\begin{figure}[htbp]
  \centering
  \includesvg[width=0.8\linewidth]{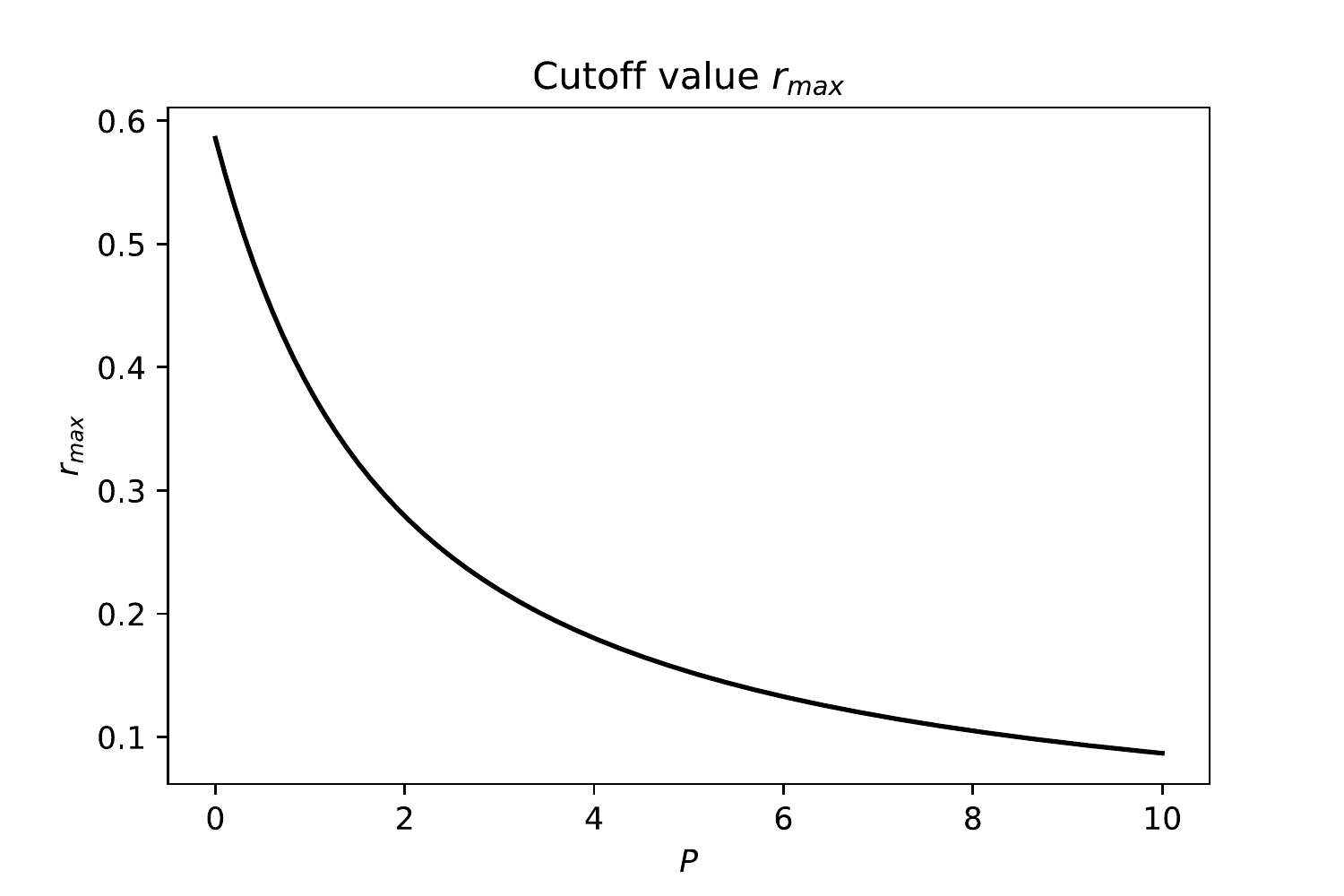}
  \caption{The cutoff goes to zero when $P \rightarrow\infty$.}
  \label{fig:cutoff}  
\end{figure}

 The behavior of the cutoff ${r_{max}}$ is displayed in Figure \ref{fig:cutoff}. Unsurprisingly, when $P \rightarrow \infty$, the penalty becomes much larger than the reward and the players play closer to $0$.

The case when $P \rightarrow 0$ is more surprising: the maximal cutoff value at $P=0$ is $h = 1 - \sqrt{\frac{\sqrt{2}}{\sqrt{2} + 2}} \approx 0.356$. This is because even if the penalty is $0$, the players cannot get the reward if they ``lose'', which prevents them from taking too much risk. We plot the distribution in Figure \ref{fig:nash}. 

\section{Generalization to multiple players}

Quite naturally, we wonder what the Nash equilibrium looks like for multiple players. The visibility game of \citet{Lotker2008-tx} probably does not admit an analytical solution and they instead give an algorithm to produce approximate solutions. We show that the CfR game for multiple players admits a unique symmetric equilibrium and present a new numerical algorithm to compute it. We finally study the asymptotic behavior of the equilibrium. 

\subsection{Nash equilibrium}
Interestingly, our Correlation for Risk game admits an analytical solution even for multiple players. More precisely:
\begin{theorem}\label{thm:multiple}
    There is a unique symmetric Nash equilibrium for in the CfR game with $n$ players defined by $$f(x) = \frac{P + w}{(n-1)(1-x)^{2+\frac{1}{n-1}} (P x + w)^{1 - \frac{1}{n-1}}}$$
    for some constants $r_{max}$ and $w := \bar r ^ {n-1}$ (the probability of winning when taking no risk) such that 
    \begin{align*}
        \int_0^{r_{max}} f(x) dx &= 1 \\
        \int_0^{r_{max}} x f(x) dx &= \bar r 
    \end{align*}
\end{theorem} 
\begin{proof}
    
We adapt the proof of Theorem \ref{thm:nashcor} and start by assuming the existence of a symmetric mixed equilibrium defined by the probability density $f$. First we derive a nice expression for $u(x)$, defined as the utility of one player choosing move $x$ while the others play according to $f$. For all $x \in support(f)$:

$$u(x) = -x P + (1-x) \left( \int_0^x f(y) dy +\int_x^1 y f(y) dy \right)^{n-1}$$

This equation is quite natural: the player loses $P$ with probability $x$. If they survive, with probability $1-x$, they need the $n-1$ other players to either play a lower value or play a higher value and fail. We can suppose as previously that $0$ is in the support to subtract $u(0)$. We write $\bar r$ for the expectation of the action $r$ under $f$.
$$\left(\frac{\bar r^{n-1} + x P}{1-x} \right)^\frac{1}{n-1} = \int_0^x f(y) dy + \int_x^1 y f(y) dy$$
We define $w := \bar r ^ {n-1}$ to be the probability of winning when taking no risk, we derivate and divide by $1-x$ to obtain:
$$f(x) = \frac{P + w}{(n-1)(1-x)^{2+\frac{1}{n-1}} (P x + w)^{1 - \frac{1}{n-1}}}$$
Finally we can solve $\int_0^{r_{max}} f(x) dx= 1$ and $\int_0^{r_{max}} x f(x) dx = \bar r$. We relegate the description of the numerical estimation of ${r_{max}}$ and $w$ to \ref{estimate}.
\end{proof}

We display the behavior of the solution for multiple players in Figure \ref{fig:solution-multiple}.

\begin{figure}[htbp]
  \centering
  \begin{minipage}[t]{0.48\textwidth}
    \centering
    \includesvg[width=1.1\linewidth]{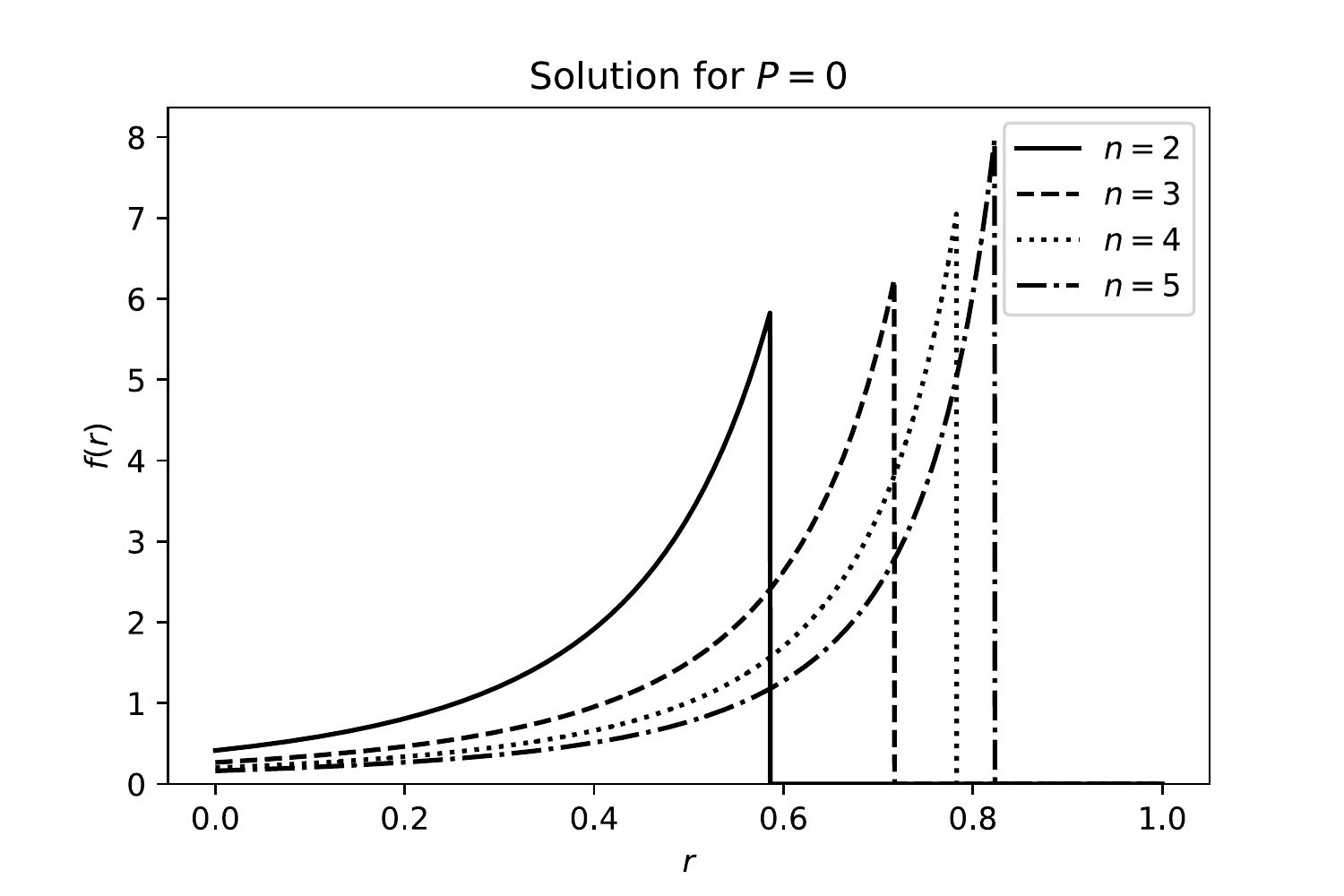}
  \end{minipage}
  \hfill
  \begin{minipage}[t]{0.48\textwidth}
    \centering
    \includesvg[width=1.1\linewidth]{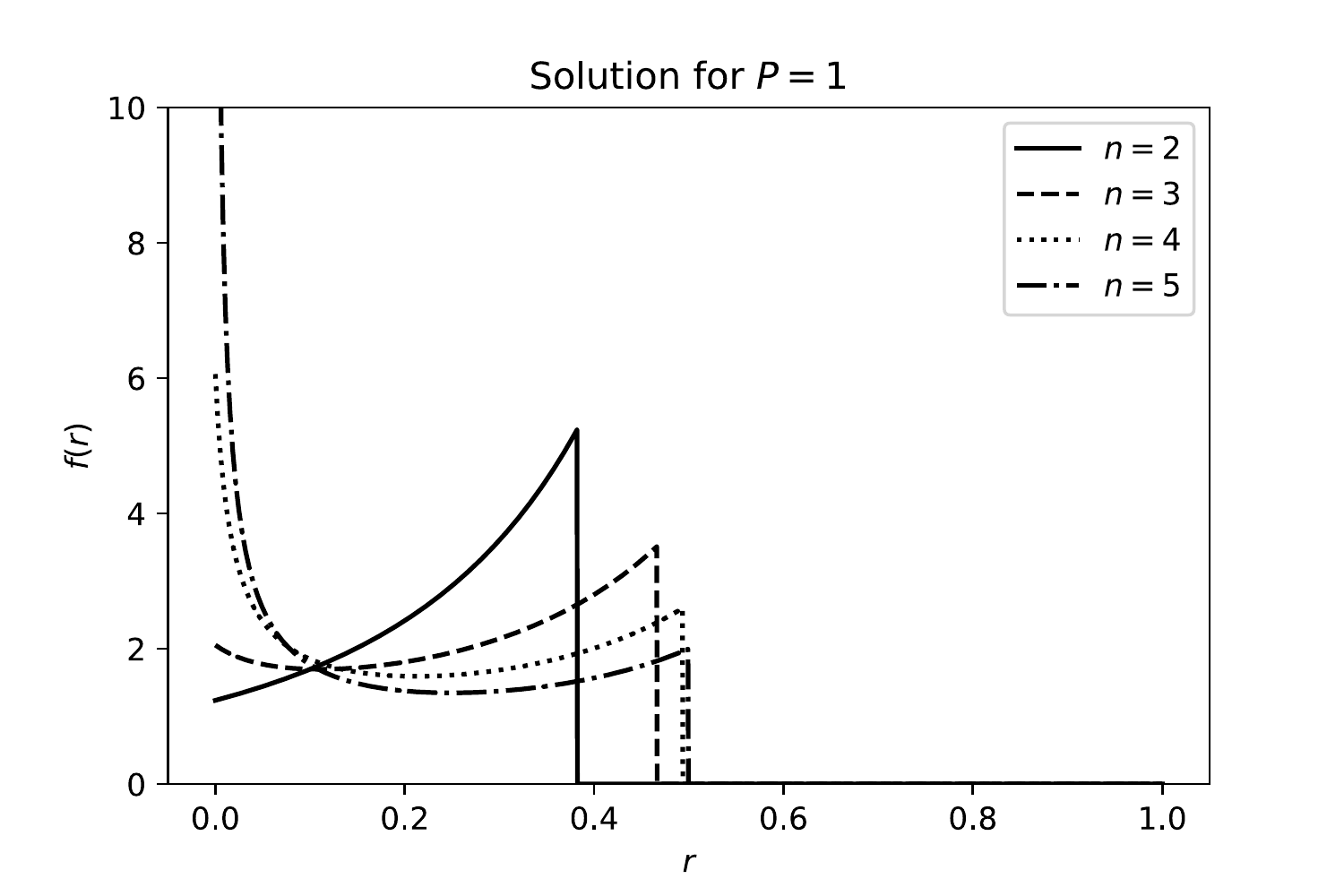}
  \end{minipage}
  \caption{We observe a clear difference between the cases $P=0$ (no penalty) and $P=1$ (presence of a penalty). In both cases, the cutoff increases. However, the average risk seems to decrease sharply when there is a nonzero penalty with a mode at $r=0$.}
  \label{fig:solution-multiple}  
\end{figure}

\subsection{Asymptotic behavior}

We are interested in studying the equilibrium when the number of players goes to infinity. For fixed $P$, we have the following:

\begin{proposition}When $n\rightarrow \infty$,
    $\lim {r_{max}} = \frac{1}{1+P}$ and $\bar r \sim \frac{1}{n P}$
\end{proposition}

\begin{proof}
We verify experimentally that $r_{max}$ is never close to 0 or 1 and that $\bar r \rightarrow 0$.
Equation \ref{eq:multi1} gives

\begin{align*}
    \frac{w + n P (1-{r_{max}}) + P {r_{max}}}{n (1-{r_{max}})(P+w)} \sqrt[n-1]{\frac{P{r_{max}}+w}{1-{r_{max}}}} &= 1 + \frac{w + nP}{n(P+w)} \bar r \\
    \sqrt[n-1]{\frac{P{r_{max}}+w}{1-{r_{max}}}} &\rightarrow 1
\end{align*}

Equation \ref{eq:multi2} gives

\begin{align*}
    \frac{w - nw(1-{r_{max}})+P{r_{max}}}{n(1-{r_{max}})(P+w)} \sqrt[n-1]{\frac{P{r_{max}}+w}{1-{r_{max}}}} &= \bar r \frac{w + nP}{n(P+w)} \\
    \frac{w}{P} + \frac{{r_{max}}}{n(1-{r_{max}})} \sim \frac{{r_{max}}}{n(1-{r_{max}})}  &\sim \bar r
\end{align*}

using $w = \bar r ^{n-1} = o(\bar r)$.

$\sqrt[n-1]{\frac{P{r_{max}}+w}{1-{r_{max}}}} \rightarrow 1$ implies $\frac{{r_{max}}}{1-{r_{max}}} \rightarrow \frac{1}P$ and ${r_{max}} \rightarrow \frac{1}{1+P}$.\\Finally, $\bar r \sim \frac{1}{nP}$.
\end{proof}

We illustrate this behavior in Figure \ref{fig:cutoff-asymptotic}.

\begin{figure}[htbp]
    \centering
    \includesvg[width=0.8\linewidth]{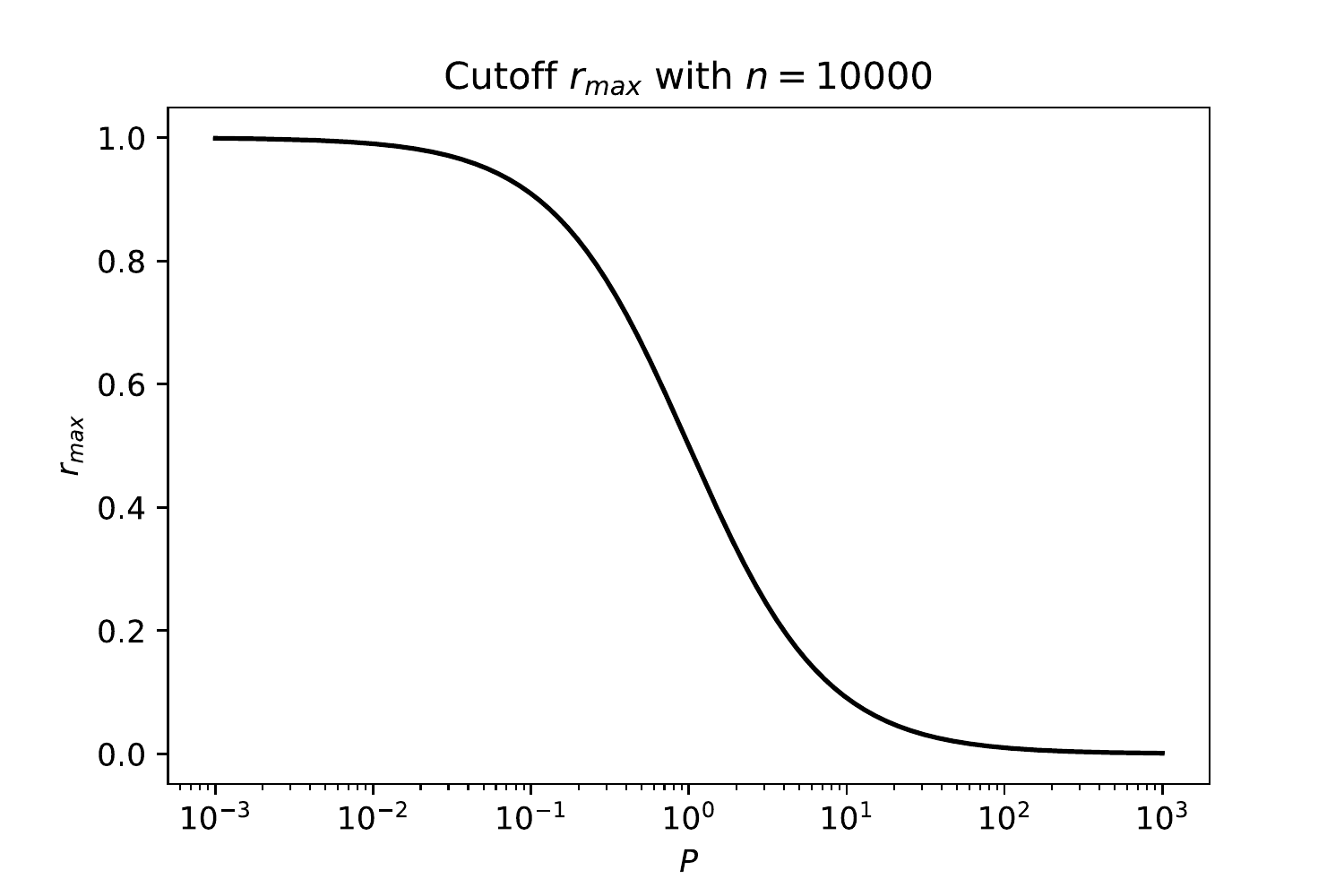}
    \caption{The plot of $r_{max}$ to $\log(P)$ is very similar to the function $\frac{1}{1 + \exp(\cdot)}$. This is because $r_{max} \sim \frac{1}{1+P}$.}
    \label{fig:cutoff-asymptotic}
\end{figure}

A common concept in game theory is the price of anarchy $PoA$ \citep{koutsoupias1999worst}. The price of anarchy is the ratio between the Pareto optimum and the Nash equilibrium. It is easy to generalize Theorem \ref{thm:pareto} for multiple players and show that the reward can be split almost perfectly to obtain a Pareto optimal utility $\frac{R}{n}$. The utility of our symmetric equilibrium is $R \bar r ^ {n-1} = R \bar r ^ {n-1} = R w$. Hence, $PoA = 1 / n w$. We will instead compute the efficiency $E = \frac{1}{PoA} = n w \in [0,1]$.

We observe that when $P = \frac{1}{n^e}$ with $e \ge 0$, the efficiency $E = n w$ of the Nash equilibrium goes to $0$ if $e \le 1$ and it goes to $1$ if $e > 1$. We plot the behavior of $E$ in Figure \ref{fig:efficiency}.
We interpret this as an indication that resources, here modeled by the ratio $\frac{1}{P} = \frac{R}{P}$ of rewards to penalties, need to scale faster than the number of players for them to adopt an efficient behavior. Scarcity of resources creates an inefficient Nash equilibrium.

\begin{figure}[htbp]
    \centering
    \includesvg[width=0.8\linewidth]{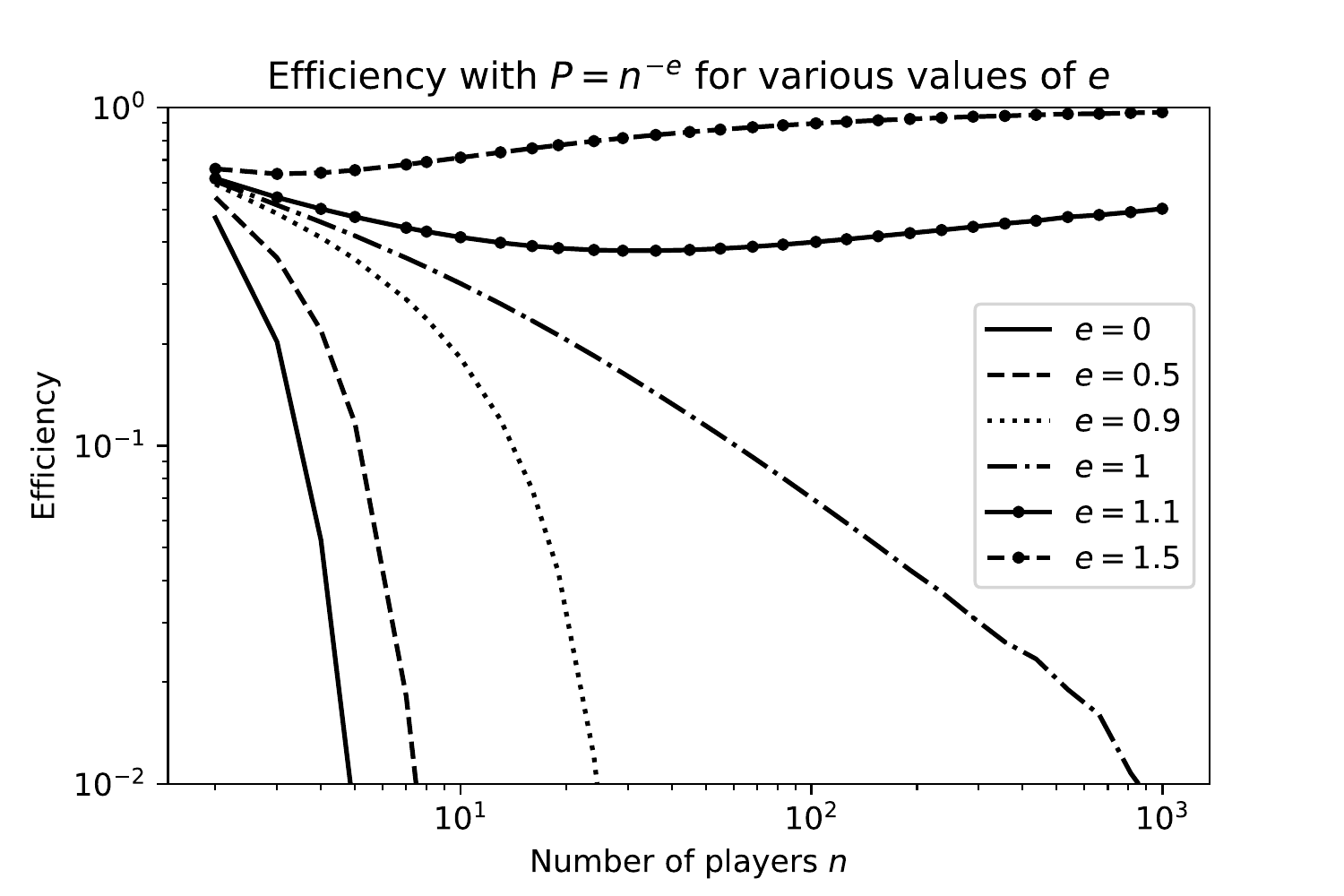}
    \caption{The efficiency clearly goes to $0$ even when $e = 1$. When $e>1$, it seems that $E \rightarrow 1$. Values for greater values of $n$ suffer of numerical precision issues as $\log w \rightarrow 0$.}
    \label{fig:efficiency}
\end{figure}

\section{Extensions of the Competition for Risk Game}

\subsection{Market frictions}

One limitation of our model is the assumption that the utility functions are discontinuous at a certain threshold. While this is appropriate for certain scenarios such as call for bids, it may not hold in other real-life situations that involve noisy evaluations or aggregate many individual choices. To address this limitation, we propose replacing the threshold $[r_1 > r_2]$ with a smooth choice model using the logistic function, $\sigma_\tau(r_1 - r_2)$, where $\sigma_\tau$ is the scaled sigmoid:

$$\sigma_\tau(x) := \frac{1}{1+\exp\left(-\frac{x}{\tau}\right)}$$

Recall that the failure events are $f_p < r_p$. For a game between two players, the outcome matrix can be represented as follows:
\begin{center}
\begin{tabular}{|c|c|c|}
\cline{2-3}
\multicolumn{1}{c|}{} & $f_1\ge r_1$ & $f_1<r_1$ \\
\hline
$f_2\ge r_2$ & $R~\sigma_\tau(r_1 - r_2), R~\sigma_\tau(r_2 - r_1) $ & $-P, R$ \\
\hline
$f_2<r_2$ & $R, -P$ &  $-P, -P$ \\
\hline
\end{tabular}
\end{center}

\begin{proposition}
\label{prop:frictions}
The expected utilities $u_p$ for the 2-player game with frictions are computed as follows:
\begin{align*}
    u_2(r_1, r_2) &= u_1(r_2, r_1) ~\text{(symmetry)}\\
    u_1(r_1, r_2) &= r_2 (1-r_1) R - r_1 P + (1-r_1)(1-r_2) \sigma_\tau(r_1 - r_2)R
\end{align*}
\end{proposition}

As $\tau \rightarrow 0$, $\sigma_\tau$ approaches the Heaviside step function and market frictions disappear.

\subsection{Correlation between risks}

In the real world, risks are often correlated, which is not accounted for in our current model. To incorporate correlation between risks, we can introduce joint distributions for the failure events $f_p$, which occur according to latent variables.

In our model, we assume that $f_p$ follows a uniform distribution. To introduce correlation between $f_1$ and $f_2$, we use the well-known NORTA (NORmal To Anything) method \citep{norta1997}. This method allows us to create a joint distribution $(f_1, f_2)$ such that the marginals are uniform distributions and the Pearson correlation between $f_1$ and $f_2$ can be set to any arbitrary value.

Following NORTA, we define $f_p = \Phi(z_p)$, where $\Phi$ is the cumulative distribution function of the Normal distribution, and

$$\begin{pmatrix}z_1\\z_2\end{pmatrix} \sim \mathcal{N}\left(\mu, \Sigma\right)$$
with $\mu = \begin{pmatrix}0\\0\end{pmatrix}$ and 
$\Sigma = \begin{pmatrix}
    1 & \rho(z_1, z_2) \\
    \rho(z_1, z_2) & 1
\end{pmatrix}$.
Here, $\rho$ is the Pearson correlation coefficient between $z_1$ and $z_2$, which determines the correlation between $f_1$ and $f_2$.
As shown in \citet{norta1997}, specifying a correlation between $z_p$ or $f_p$ is equivalent to specifying $\rho(f_1, f_2)$. Specifically, we have $$\rho(f_1, f_2) = \frac{6}{\pi}\sin^{-1}\left(\frac{\rho(z_1, z_2)}{2}\right)$$ 
Hence, we use $\rho$ to denote $\rho(z_1, z_2)$ throughout the rest of the document.
This model is well-suited to real-world scenarios, such as financial portfolios, where $z_p$ can represent the returns on investments. In such cases, joint distributions of portfolios are typically modeled as multivariate normal distributions, and $r_p$ corresponds to the Value at Risk $v_p = \Phi^{-1}(r_p)$ through the bijective function $\Phi$, such that the failure event $z_p < v_p$ is equivalent to $f_p < r_p$.

\begin{proposition}
    The expected utilities $u_p$ for the 2-player game with frictions and correlated risks are computed as follows:

    \begin{align}
    \label{eq:utility}
    u_2(r_1, r_2) &= u_1(r_2, r_1) ~\text{(symmetry)}\\
    u_1(r_1, r_2) &= (r_2 - \tilde r) R - r_1 P + (1-r_1-r_2 + \tilde r) \sigma_\tau(r_1 - r_2)R
    \end{align}

    where $\tilde r := \Phi_\rho(\Phi^{-1}(r_1), \Phi^{-1}(r_2))$ is the probability of joint failure, with $$\Phi_\rho(v_1, v_2) = \frac{1}{2\pi\sqrt{1 - \rho^2}} \int_{-\infty}^{v_1}\int_{-\infty}^{v_2} \exp\left(-\frac{x^2 - 2\rho x y + y^2}{2(1-\rho)^2}\right) dy~dx$$ the cumulative distribution of the bivariate normal distribution with correlation $\rho$.
\end{proposition}

In the absence of noise, when $\rho = \pm 1$, it is also possible to calculate the Nash equilibrium analytically:

\begin{restatable}{theorem}{nashwithcor}
\label{thm:nashwithcor}
        For $\rho = 1$, the equilibrium is given by:

    $$p(x) = \frac{1+P}{1-x} \left[x < 1 - \exp\left(-\frac{1}{P+1}\right)\right]$$ 
    We have $$\bar r = 1 - (P+1)\left(1-\exp\left(-\frac{1}{P+1}\right)\right)$$

    For $\rho = -1$, the equilibrium is given by:

    $$p(x) = \frac{P}{(1-2x)^{3/2}} \left[x < \frac{1}{2} - \frac{P^2}{2(P+1)^2}\right]$$
    We have  $$\bar r = \frac{1}{2P+2}$$

\end{restatable}
\begin{proof}
    See \ref{proof:nashwithcor}.    
\end{proof}
\section{Computing approximate Nash equilibrium}

\subsection{Approximations to games and equilibria}
\label{sec:approx}

In this section, we define some key concepts and metrics related to games and equilibria. 

For a given game with $n$ players, we use $u_i(\sigma)$ to denote the reward of player $i$ when all players follow the strategy $\sigma = (\sigma_1, \ldots, \sigma_n)$\footnote{This $\sigma$ is not to be confused with the scaled sigmoid $\sigma_\tau$ defined earlier.}. A strategy $\sigma$ is said to be a Nash equilibrium if it satisfies the following condition for all players $i$ and all alternative strategies $\sigma_i' \in \Sigma_i$: 
$$u_i(\sigma) \ge u_i(\sigma_i', \sigma_{-i})$$ 
where $\sigma_{-i}$ is the strategy of all players but $i$, and $\Sigma_i$ is the set of actions available to player $i$. 

A game is said to be continuous if the action space $\Sigma_i$ is compact and $u_i$ is continuous. In such games, it is possible to approximate the Nash equilibria using a sequence of games over a reduced finite support, which leads to Glicksberg’s theorem, without relying on Kakutani’s theorem \citep{myerson1997game}. 

In our CfR game, which has a few points of discontinuity, it is also possible to approximate the Nash equilibria using a similar method. However, we do not provide a proof of this here, as the introduction of frictions makes our game continuous anyway. 

To measure the closeness of a strategy $\sigma$ to a Nash equilibrium, we use the NashConv metric \citep{Lanctot2017-wc}: 
$$\textsc{NashConv}(\sigma) = \sum_{i=1}^n \max_{s_i \in \Sigma_i} u_i(s_i, \sigma_{-i}) - u_i(\sigma)$$
Note that this metric only considers pure strategies $s_i \in \Sigma_i$, due to the linearity of the payoff function for mixed strategies. 

The \textsc{NashConv} metric satisfies $\textsc{NashConv}(\sigma) \ge 0$, with equality holding only for a Nash equilibrium. This implies that $\textsc{NashConv}(\sigma)$ corresponds to the notion of $\varepsilon$-Nash equilibrium, where a $\varepsilon$-Nash equilibrium $\sigma$ has $\textsc{NashConv}(\sigma) = n \varepsilon$. 

For a finite action space, $\textsc{NashConv}$ is easy to compute since $\Sigma_i$ is finite. However, for a continuous action space, no such metric is known. Nonetheless, we can approximate $\textsc{NashConv}$ by taking the maximum over a finite sample of points from $\Sigma_i$. This sample can be chosen randomly, or if $\Sigma_i$ is an interval or a product of intervals of $\mathbb{R}$, we can use a grid. 

In our CfR game, the action space is $[0,1]$. Here, we use quasi-random numbers to measure the closeness to a Nash equilibrium, inspired by the literature on hyperparameter sampling \citep{Bousquet2017-kg} and the efficiency of quasi-Monte-Carlo methods \citep{sobol1990quasi}. Specifically, we define the QuasiNashConv metric as: 
$$\textsc{QuasiNashConv}(\sigma, m) = \sum_{i=1}^n \max_{s_i \in \textsc{Sobol}(m)} u_i(s_i, \sigma_{-i}) - u_i(\sigma)$$

where $\textsc{Sobol}(m)$ is a set of $m$ quasi random numbers drawn using Sobol’s method \citep{sobol1967distribution}.

\subsection{Correlated Equilibria}

A Nash equilibrium is a set of strategies where no player can improve their payoff by unilaterally changing their strategy, assuming that all other players' strategies remain unchanged. However, in some games, players may benefit from coordinating their actions in ways not captured by traditional Nash equilibrium. This is where the concept of correlated equilibrium comes in.

A correlated Nash equilibrium is a set of correlated strategies where no player can improve their expected payoff by unilaterally changing their strategy, given that they observe the correlation signal. This correlation signal is not necessarily a message or communication between the players, but rather a shared random variable that affects each player's strategy consistently.

\begin{definition}
\label{def:cor}
A correlated Nash equilibrium is a joint distribution $\sigma$ over all moves $\Sigma_1 \times \Sigma_2 \times \ldots \times \Sigma_n$ such that for any player $i$ and any strategy modification $\phi: \Sigma_i \rightarrow \Sigma_i$, $$u_i(\sigma_i, \sigma_{-i}) \ge u_i(\phi(\sigma_i), \sigma_{-i})$$
\end{definition}

Thus, a Nash equilibrium can be viewed as a correlated Nash equilibrium that can be decomposed into independent strategies for each player. It is evident that any Nash equilibrium is a correlated Nash equilibrium.

Correlated equilibria are more suitable for the real world because they allow for a broader range of possible outcomes that can arise through coordination among the players, without necessarily requiring communication or binding agreements between them.

In many real-world scenarios, it is challenging or impossible for players to communicate and make binding agreements, or they may not have complete information about the strategies of the other players. Correlated equilibria provide a way for players to achieve coordination and cooperation without requiring such communication or information, by relying on shared random variables that affect each player's strategies consistently.

Finally, correlated equilibria can also capture situations where players have some degree of trust or social norms that encourage them to coordinate their actions in a specific way. For instance, in a repeated game where players interact with each other over a long period, they may develop a sense of reciprocity or reputation that encourages them to follow a certain coordinated strategy.

\subsection{Finding Correlated Equilibria with Linear Solvers}
\label{sec:linear}
Correlated equilibria are of interest because they can be computed more easily for a finite action set.

A joint strategy can be represented by a mapping of probabilities: $$\Pr\nolimits_\sigma(s_1, s_2, \ldots, s_n) := \Pr[\sigma = (s_1, s_2, \ldots, s_n)]$$ for all joint actions $(s_1, s_2, \ldots, s_n)$. Therefore, the equation from Definition \ref{def:cor} is linear in these probabilities. An additional equation is that probabilities must sum to 1, and all probabilities are constrained to be positive. For two players, the equations are:

$$\forall (s_1, s'_1), \sum_{s_2} \Pr\nolimits_\sigma(s_1, s_2) u_1(s_1, s_2) \ge \sum_{s_2} \Pr\nolimits_\sigma(s_1, s_2) u_1(s'_1, s_2)$$

$$\forall (s_2, s'_2), \sum_{s_1} \Pr\nolimits_\sigma(s_1, s_2) u_1(s_1, s_2) \ge \sum_{s_1} \Pr\nolimits_\sigma(s_1, s_2) u_1(s_1, s'_2)$$

$$\forall (s_1, s_2), \Pr\nolimits_\sigma(s_1, s_2) \ge 0$$

$$\sum_{s_1,s_2} \Pr\nolimits_\sigma(s_1, s_2) = 1$$

The set of correlated equilibria is thus a convex polytope $P$. It is possible to find the boundary in any direction using a linear programming solver.

It is also possible to check that the correlated equilibrium is unique and is a Nash equilibrium by trying to maximize and minimize each variable over the polytope. If the maximum and minimum are equal for each variable, then the polytope only contains one point. Another method described in \citet{appa2002uniqueness} checks the uniqueness of a solution to a linear program by solving a new linear program. However, that method requires a reformulation of the linear program as $\max cx \text{ s.t. } Ax=b, x \ge 0$, which is cumbersome in our case. We propose a simple randomized method (algorithm \ref{alg:diameters}) that can produce confidence intervals for any confidence level (or p-value).

\begin{restatable}{theorem}{pvalue}
\label{thm:pvalue}
Given a polytope $P$ defined by constraints $c_1, \ldots, c_m$
 $$
        \Pr\left[\textsc{SumDiamSquared}(K, c_1, \ldots, c_m) < \varepsilon\right] \le F_{\chi^2}\left(\frac{\varepsilon}{diam(P)}, K\right)
   $$

   with $F_{\chi^2}(\cdot, K)$ the cumulative distribution function of the $\chi^2$ distribution with $K$ degrees of freedom.
\end{restatable}

\begin{algorithm}
\label{alg:diameters}
\caption{Confidence interval on $diam(P)$}
\begin{algorithmic}
\Require{Iterations $K$, constraints $c_1, \ldots, c_m$ defining a polytope $P$ in $\mathbb{R}^n$}
\Function{SumDiamSquared}{$K, c_1, \ldots, c_m$}
\For{$i\gets 1,\ldots,K$}
    \State Sample $v_j \sim \mathcal{N}(0, 1)$ for $j=1,\ldots,n$
    \State $a_i \gets \textsc{LinProg}(v, c)$ \Comment{$\min \limits_{x\in P}v \cdot x$}
    \State $b_i \gets \textsc{LinProg}(-v, c)$ \Comment{$\max \limits_{x\in P}v \cdot x$}
    \State $d_i \gets b_i - a_i$
\EndFor
\State \Return $\sum_i d_i^2$
\EndFunction
\Require{p-value $p$}

\Function{MaxDiameter}{$p, K, c_1, \ldots, c_m$}
\State $\varepsilon \gets$ \Call{SumDiamSquared}{$K, c_1, \ldots, c_m$}
\State $q \gets$ \Call{Chi2.ppf}{$p, K$}
\State $d \gets \varepsilon / q$
\State \Return $d$
\EndFunction
\end{algorithmic}
\end{algorithm}

We used the HiGHS solver \citep{huangfu2018parallelizing} to solve the linear optimization subproblems (calls to \textsc{LinProg}). In numerical experiments, we use $K=5$, confidence $p=0.95$, and report $$d_{max} :=  \textsc{MaxDiameter}(p, K, c_1, \ldots, c_m) = \frac{\textsc{SumDiamSquared}(K, c_1, \ldots, c_m)} {Q_{\chi^2}\left(1-p, K\right)}$$ where $Q_{\chi^2}(\cdot, K)$ is the quantile function of the $\chi^2$ distribution with $K$ degrees of freedom. When the polytope describe probability distributions, we have the bound $d_{max} \leq 2$.

Finally, we make the following trivial remark:

\begin{proposition}
    A correlated equilibrium $\sigma$ is a Nash equilibrium iff the matrix $(\Pr\nolimits_\sigma(i, j))_{i,j}$ has rank 1.
\end{proposition}

This gives us another numerical method to check that a correlated equilibrium is a Nash equilibrium: compute the second highest eigenvalue and check that it is $0$. In numerical experiments, we define $\lambda_1$ and $\lambda_2$ as the highest and second highest eigenvalues and report the value $$\lambda := \frac{\lambda_2}{\lambda_1}$$

\subsection{Related works solving continuous games}

Our exploration of regret-minimization algorithms in the context of continuous games has led us to a wide range of approaches. From these, a few distinct groups emerge, each characterized by their unique methods, assumptions, and requirements.

One group includes the work of \citet{perkins2014stochastic}, \citet{ganzfried2021algorithm}, and \citet{kroupa2023multiple}. \citet{perkins2014stochastic} extended stochastic fictitious play to the continuous action space framework, showing convergence to an equilibrium point in two-player zero-sum games. However, this method assumes specific linear or quadratic utility functions, which limits its applicability to a narrow set of games. Similarly, \citet{ganzfried2021algorithm} proposed a novel algorithm for approximating Nash equilibria in continuous games, yet the scalability of their method is a concern due to the storage of all previous moves and the use of mixed integer linear programs for best response computation. Lastly, \citet{kroupa2023multiple} presented an iterative strategy generation technique for finding mixed strategy equilibria in multiplayer general-sum continuous games, but this method requires an oracle for best response computation, a requirement that may not be met in all practical scenarios.

Another group consists of \citet{raghunathan2019game} and \citet{dou2019finding}. \citet{raghunathan2019game} introduced the Gradient-based Nikaido-Isoda function, a merit function providing error bounds to a stationary Nash point. They showed that gradient descent converges sublinearly to a first-order stationary point of this function, making it a potential method for steady convergence towards equilibrium. Extending this work, \citet{dou2019finding} offered a deep learning-based approach for approximating Nash equilibria in continuous games, allowing the finding of mixed equilibria. They utilized the pushforward measure technique to represent mixed strategies in continuous spaces and applied gradient descent for convergence to a stationary Nash equilibrium. However, we found this method to be slow and unable to converge in our game.

The third group encompasses the works of \citet{bichler2021learning} and \citet{martin2022finding}. \citet{bichler2021learning} proposed using artificial neural networks to learn equilibria in symmetric auction games through gradient dynamics in self-play. This work was later extended by \citet{martin2022finding}, who used the same pushforward trick as \citet{dou2019finding}, adding random noise as input. They applied zeroth-order optimization techniques to compute approximate Nash equilibria in continuous-action games without access to gradients. However, we found that the success of this method is very sensitive to hyperparameters, and its dynamics can be unstable.

\subsection{Applying regret-minimization algorithms to continuous games}

Our study delves into the application of regret-minimization algorithms to continuous games, a distinct approach given that it does not demand any differentiability assumptions on the reward function. This approach pivots around algorithms known for identifying correlated equilibria, such as regret matching \citep{Hart1997-nk,hart2001reinforcement, hart2001reinforcementcorrection}, Counterfactual Regret Minimization (CFR) \citep{neller2013introduction}, and stochastic fictitious play \citep{fudenberg1993learning}.

Regret matching is a notable algorithm in the realm of game theory for finding correlated equilibria. It operates by having each player select a distribution of moves at each step of the algorithm, wherein this selection is geared towards maximizing their expected utility, given the past moves of their opponent. When run for an ample number of iterations, the distributions arrive at a convergence, forming a correlated equilibrium.

Building on this, we bring into play the CFR algorithm. This algorithm accumulates the expected regrets of each action played by a player at every information set. It then updates the regrets based on the counterfactual outcomes of the game - that is, what the result would have been if a different action had been taken. As these regrets are iteratively updated and actions chosen based on these revised regrets, CFR converges to a Nash equilibrium in extensive form games. In our normal form game, we incorporate CFR as a deterministic variant of Regret Matching.

Further enriching our methodology is stochastic fictitious play, another variant of regret matching. It involves the computation of probability distributions using the softmax function. We also test this variant along with CFR in our experiments, thereby evaluating their performance in approximating the Nash equilibrium of our game.

To approximate the CfR game, we resort to a finite grid of actions, which are evenly distributed within the interval $[0, 1]$ for both players. In doing so, we differentiate between two settings, governed by a boolean variable, \texttt{shift}. When \texttt{shift = false}, both players are offered the same set of actions. However, when \texttt{shift = true}, the two players are presented with non-intersecting sets of actions, with each action in the interval $[0, 1]$ alternately attributed to each player.

Despite the fact that the action set of the resulting approximate Nash equilibrium is necessarily a subset of a predefined finite action set, our method brings a fresh perspective to the field, with its emphasis on regret minimization for both players and the removal of any differentiability assumption on the reward function. Through this research, we hope to enrich our understanding of the behavior of continuous games and lay a sturdy foundation for future exploration and improvements in this domain.

\section{Experimental results}

We employ a Numba \citep{lam2015numba}  implementation of the numerical method proposed by \citet{genz2004numerical} to compute the bivariate normal probability in the utility function, which is the computational bottleneck in equation \ref{eq:utility}.

Our experiments show that the Nash equilibrium in finite approximations of the CfR game can be obtained by computing a correlated Nash equilibrium. Moreover, we prove that either the correlated equilibrium is unique or the correlated equilibrium maximizing the total reward is a Nash equilibrium.

We present the results of our experiments with different values of $P$, $\tau$, and $\rho$, as well as different numbers of actions, in Figure \ref{fig:linear-shift}, where we plot the maximum distance $d_{max}$ and the parameter $\lambda$ computed according to Section \ref{sec:linear} with \texttt{shift = true}. We also present the results with \texttt{shift = false} in Section \ref{sec:linear-noshift}.

\begin{figure}[htbp]
  \centering
  \begin{minipage}[t]{0.48\textwidth}
    \centering
    \includesvg[width=1.1\linewidth]{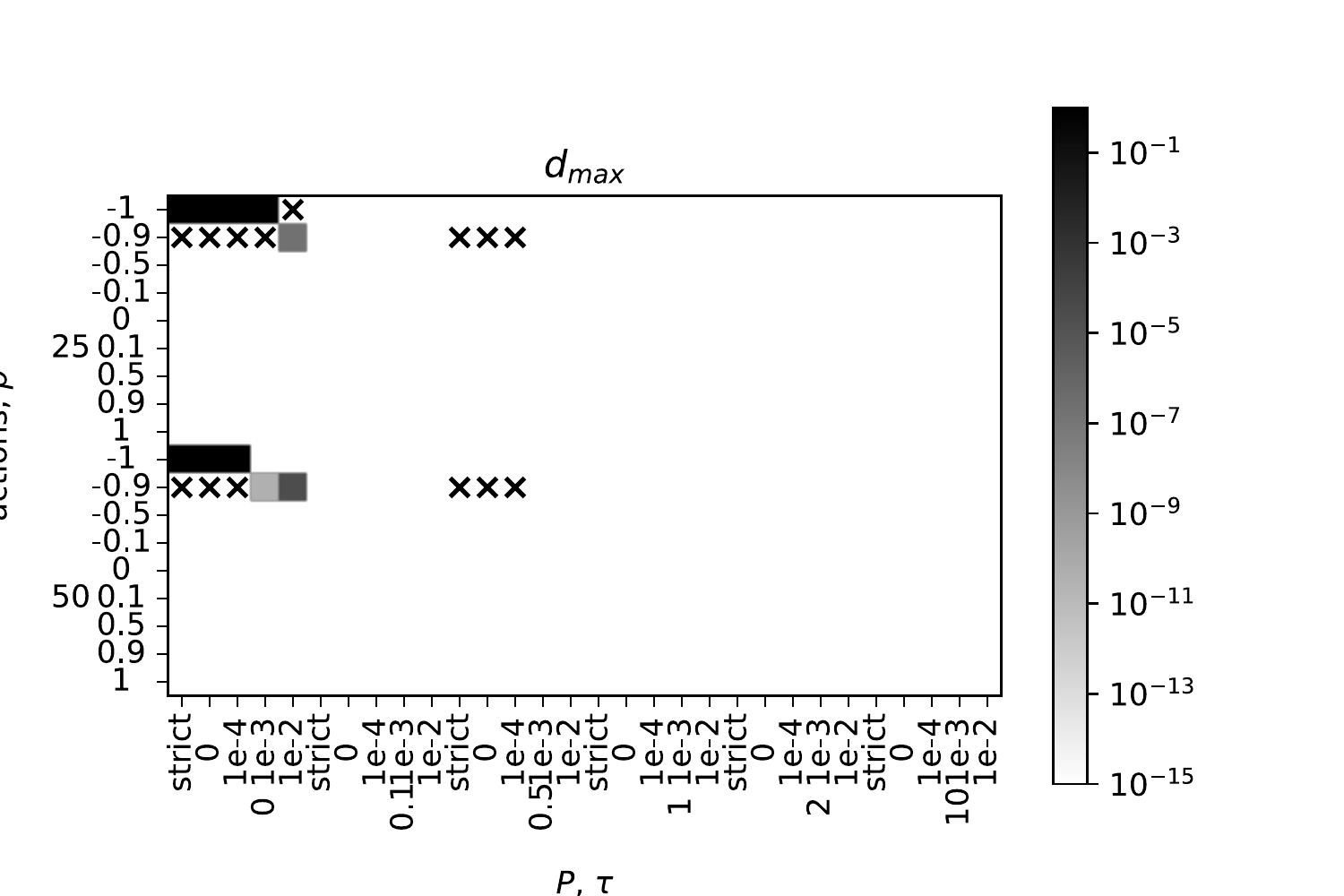}
  \end{minipage}
  \hfill
  \begin{minipage}[t]{0.48\textwidth}
    \centering
    \includesvg[width=1.1\linewidth]{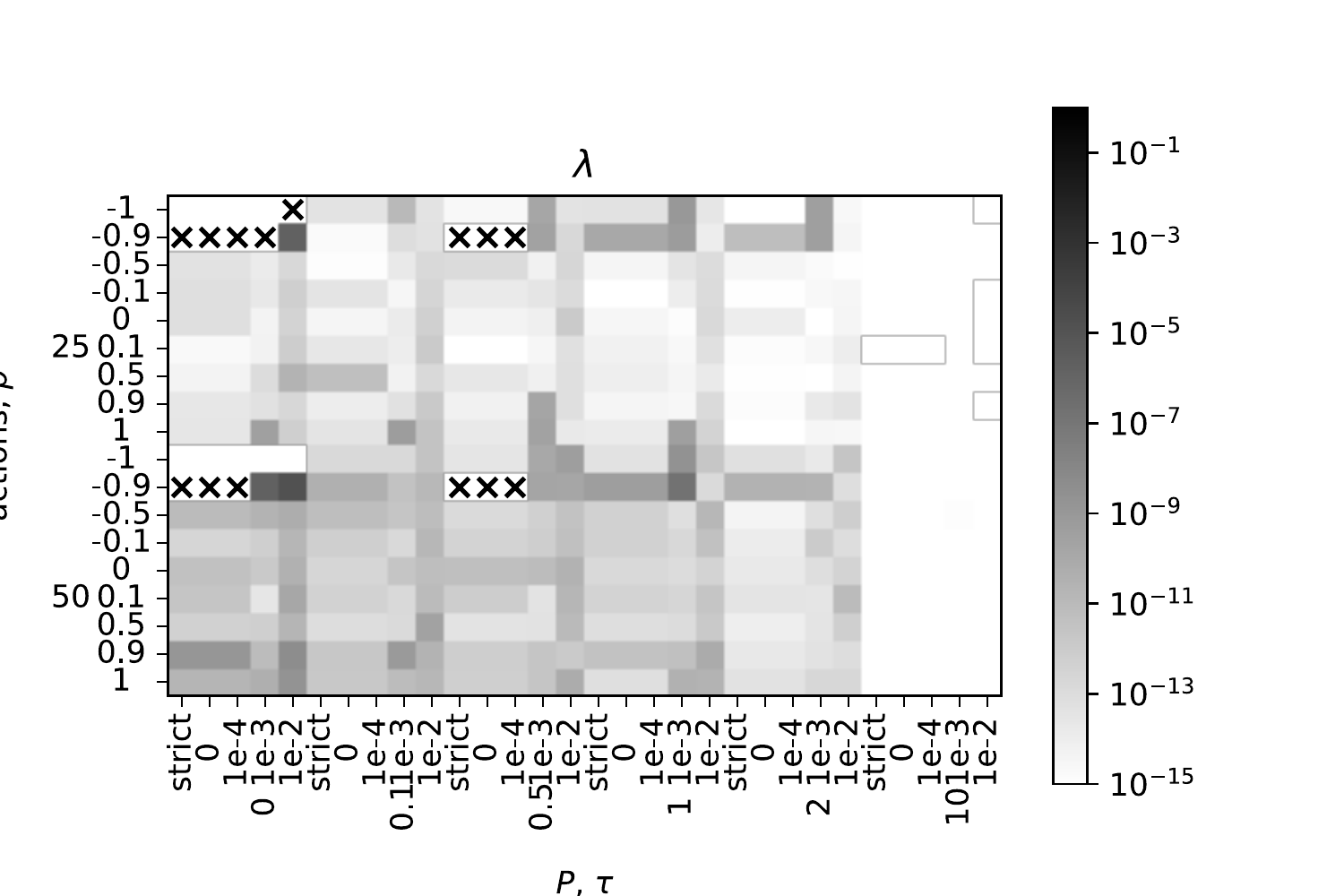}
  \end{minipage}
  \caption{Our experiments with different values of $P$, $\tau$, and $\rho$, as well as different numbers of actions, for the setting with \texttt{shift = true} show that, for almost all values of the parameters, $d_{max}$ is numerically zero. We cross out cases where the solver failed to find a solution, most probably due to rounding errors which caused the small solution set to disappear. In the few cases where $d_{max}$ was not observed to be zero, $\lambda$ is clearly zero, indicating that the best correlated equilibrium is a Nash equilibrium.}
  \label{fig:linear-shift}  
\end{figure}

We implement and evaluate all algorithms to solve a discrete version of our game where the action space is reduced to a grid, with and without a shift. We then report the value of $\textsc{NashConv}$ $\textsc{QuasiNashConv}$ as defined in Section \ref{sec:approx}. We display the solution found by Regret Matching in Figure \ref{fig:rm-sol} and check that it matches the actual solution from Figure \ref{fig:nash}. We observe on Figure \ref{fig:rm-different} that vanilla Regret Matching outperforms all methods. In general, Regret Matching performs much better than CFR, even when using less actions, using softmax is detrimental, and shifting the action space does not seem to impact performance much. CFR runs faster as it does not involve any random sampling. Therefore, we use vanilla Regret Matching with shifting in the rest of the experiments for the sake of simplicity. Figure \ref{fig:rm-size} shows that the number of sampled actions is the main factor that drives the quality of solutions.

\begin{figure}[htbp]
  \centering
  \includesvg[width=0.8\linewidth]{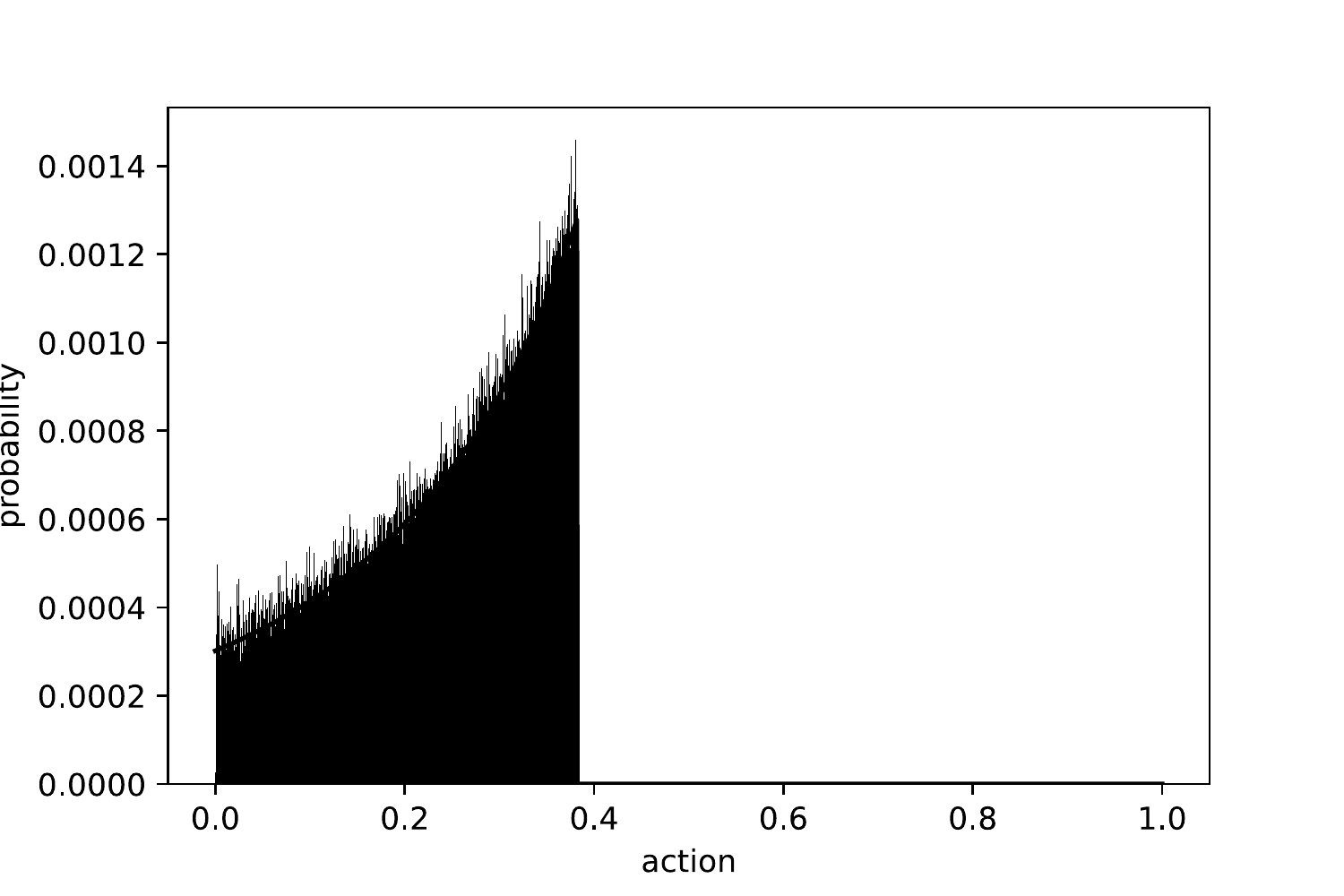}
  \caption{Approximate and analytical solutions for $P=1, \rho=0, \tau=0$. Actions were discretized to a grid of $2^{12}$ with \texttt{shift~=~true}. The approximate solution was computed by $10^4$ iterations of regret matching.}
  \label{fig:rm-sol}  
\end{figure}

For our experiments on equilibria in various settings, we use 500 actions and 2000 iterations (hence $10^6$ steps) of vanilla Regret Matching with shifting. Our efficient implementation computes the equilibrium in just a few seconds, allowing us to explore the effects of penalties, market frictions, and correlation on risk-taking behavior and performance, as well as to evaluate the effectiveness of various interventions and policies in a competitive environment.

\begin{figure}[htbp]
  \centering
  \begin{minipage}[t]{0.48\textwidth}
    \centering
    \includesvg[width=1.1\linewidth]{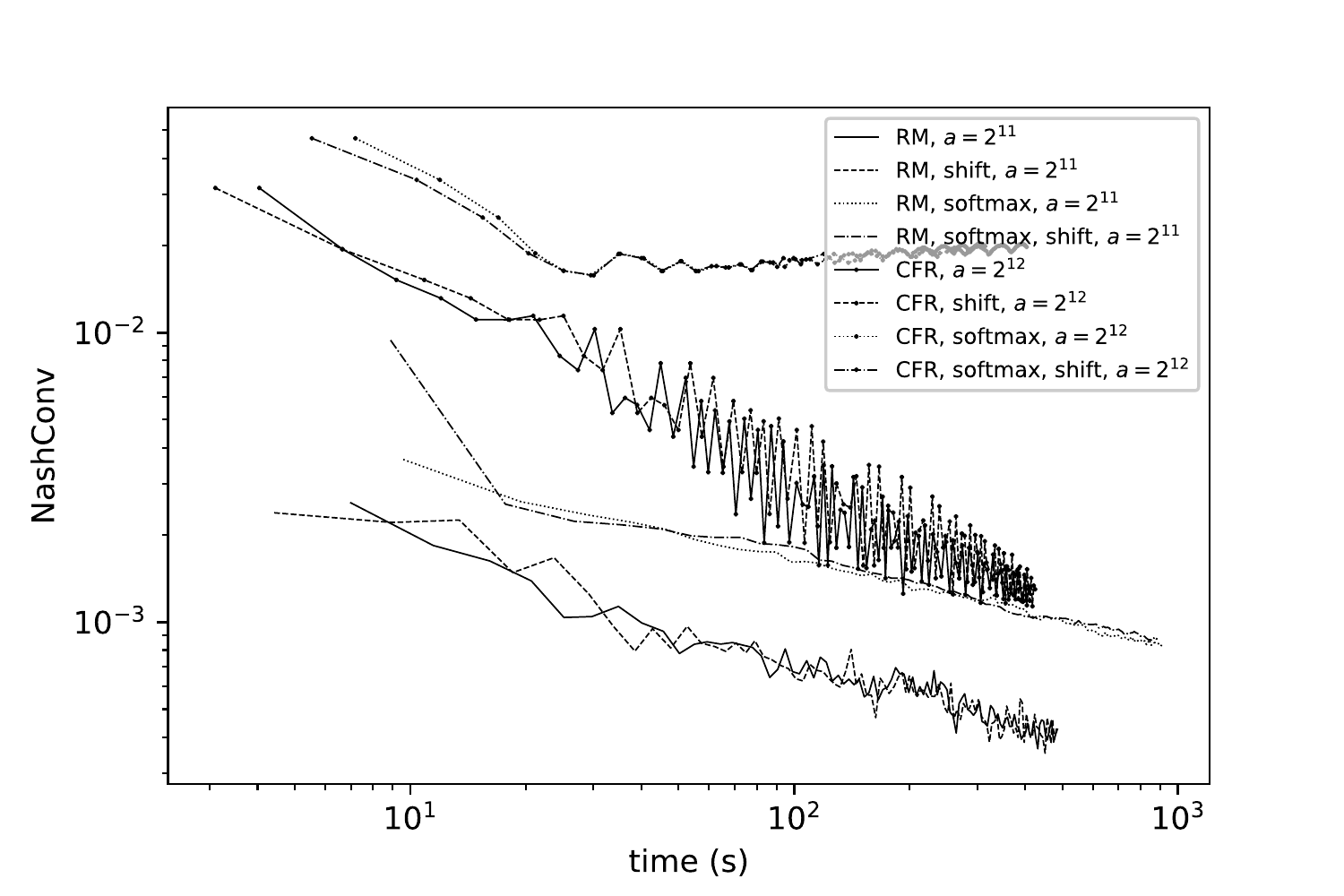}
  \end{minipage}
  \hfill
  \begin{minipage}[t]{0.48\textwidth}
    \centering
    \includesvg[width=1.1\linewidth]{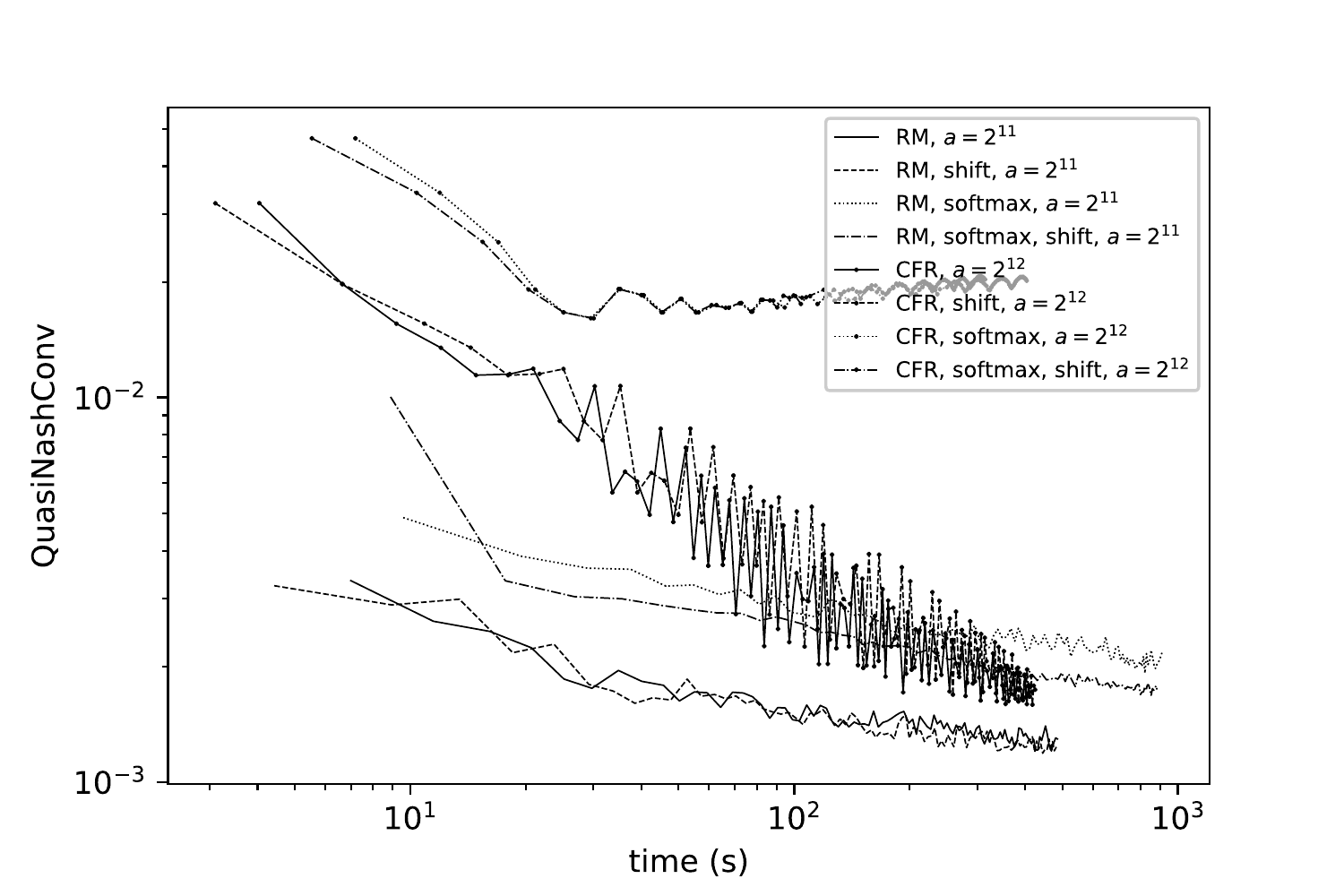}
  \end{minipage}
  \caption{We evaluate the performance of different algorithms in finding the Nash equilibrium of the CfR game in the standard setting with $P=R=1$, $\tau=0$, and $\rho=0$, without sharing. We reduce the game to a grid of $a=2^{12}$ actions, and evaluate the $\textsc{QuasiNashConv}$ metric on $2^{15}$ points. The algorithms are run for $10^4$ iterations, where an iteration is defined as an update to the strategy. For CFR, an iteration involves updating the regrets based on the counterfactual outcomes of the game, while for regret matching, an iteration is defined as $a$ steps of the sampling, play, and update process. Both CFR and regret matching do $\mathcal{O}(a^2)$ operations per iteration. We do not include the computation of the utility matrix in the computation time.}
  \label{fig:rm-different}  
\end{figure}

\begin{figure}[htbp]
  \centering
  \begin{minipage}[t]{0.48\textwidth}
    \centering
    \includesvg[width=1.1\linewidth]{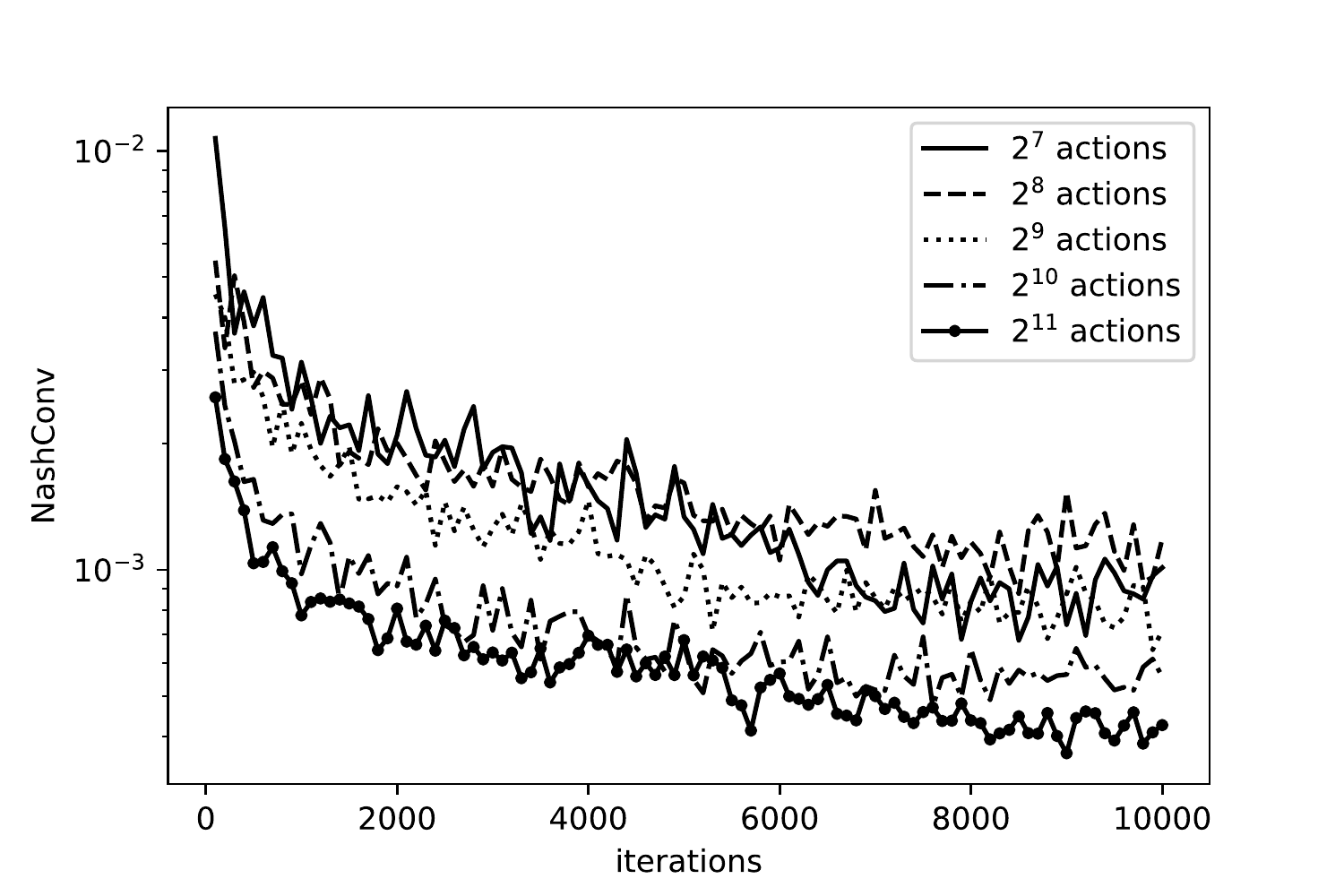}
  \end{minipage}
  \hfill
  \begin{minipage}[t]{0.48\textwidth}
    \centering
    \includesvg[width=1.1\linewidth]{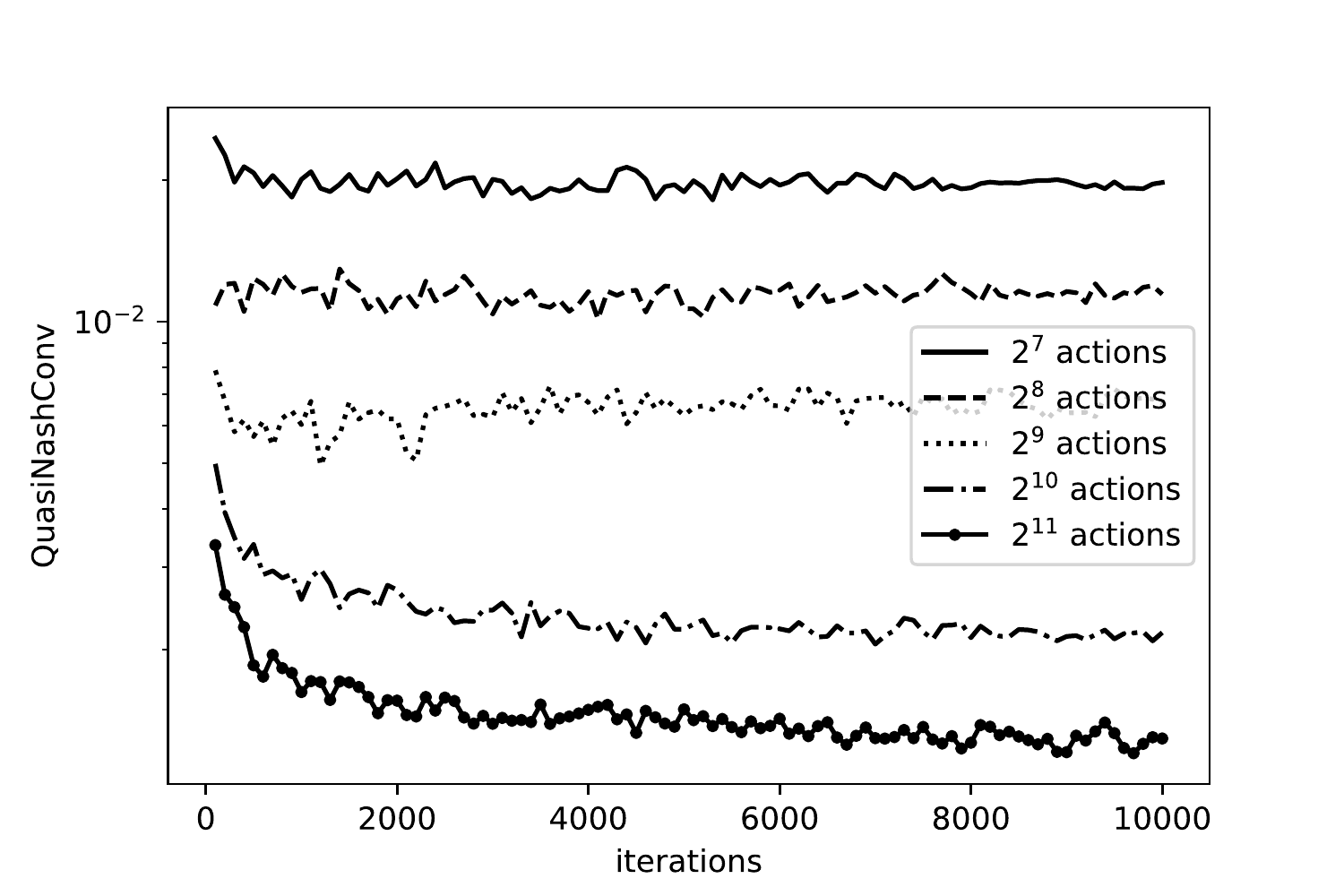}
  \end{minipage}
  \caption{We evaluate the performance of the vanilla regret matching algorithm in finding the Nash equilibrium of the CfR game in the standard setting with $P=R=1$, $\tau=0$, and $\rho=0$, without sharing. The quality of the resulting equilibrium depends mainly on the number of actions $a$. We run the algorithm for $t$ iterations, where the time complexity of the algorithm is $\mathcal{O}(t a^2)$ and the memory complexity is $\mathcal{O}(a^2)$ as it requires computing the reward matrix. We evaluate the $\textsc{QuasiNashConv}$ metric on $8a$ points.}
  \label{fig:rm-size}  
\end{figure}

Our results demonstrate that penalties and market frictions have a significant impact on the strategic behavior of actors in competition. Specifically, we find that penalties $P$ decrease both the average risk taken $\bar r$ by the players and their total reward $u$, while market frictions $\tau$ not only decrease the average risk taken but also increase the total reward. We also find that market frictions have a greater impact on the total reward in high-penalty environments. In particular, in particularly inefficient markets with high $\tau$, increasing penalties can improve cooperation and the total reward, as illustrated in Figure \ref{fig:result-noise}.

\begin{figure}[htbp]
  \centering
  \begin{minipage}[t]{0.48\textwidth}
    \centering
    \includesvg[width=1.1\linewidth]{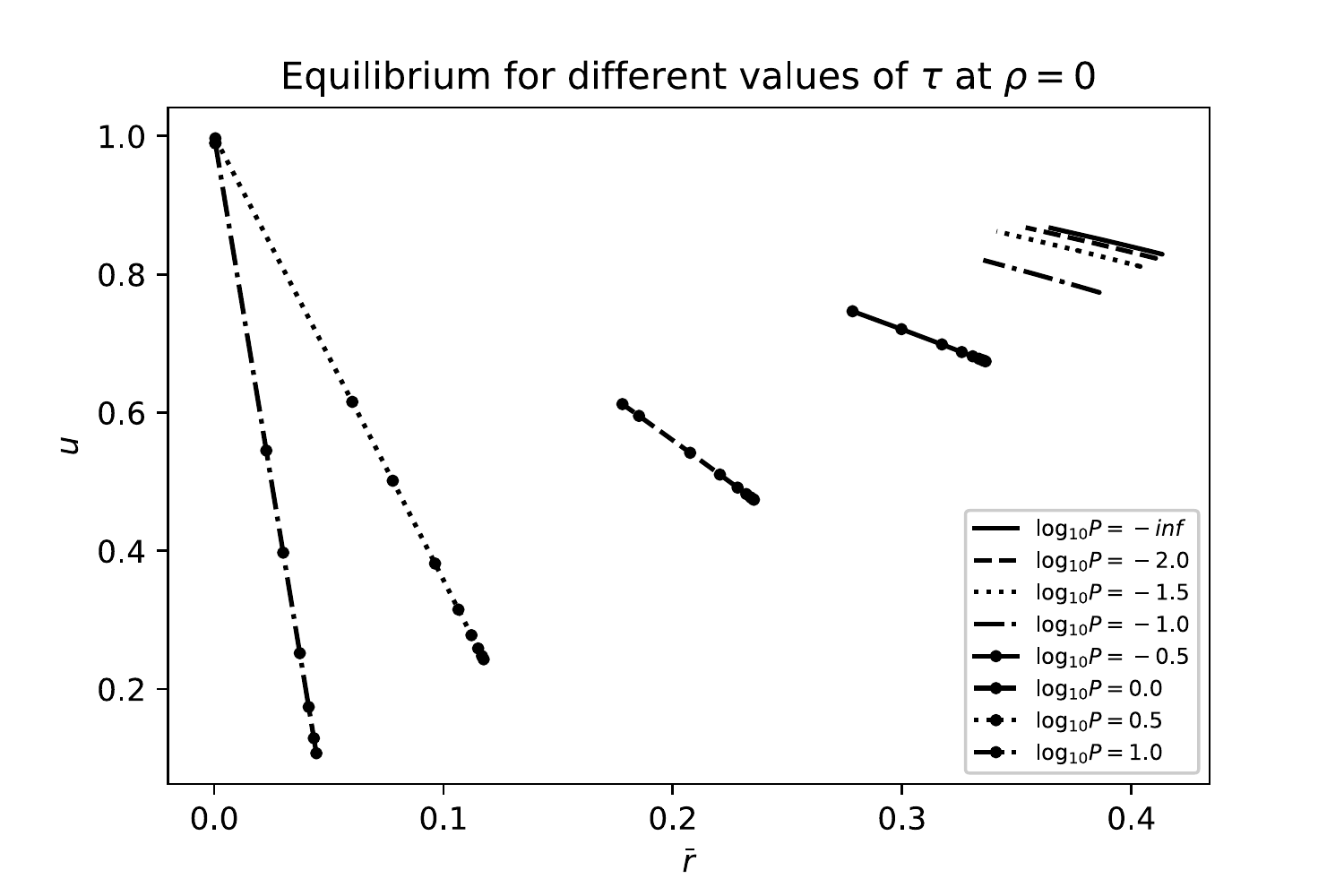}
  \end{minipage}
  \hfill
  \begin{minipage}[t]{0.48\textwidth}
    \centering
    \includesvg[width=1.1\linewidth]{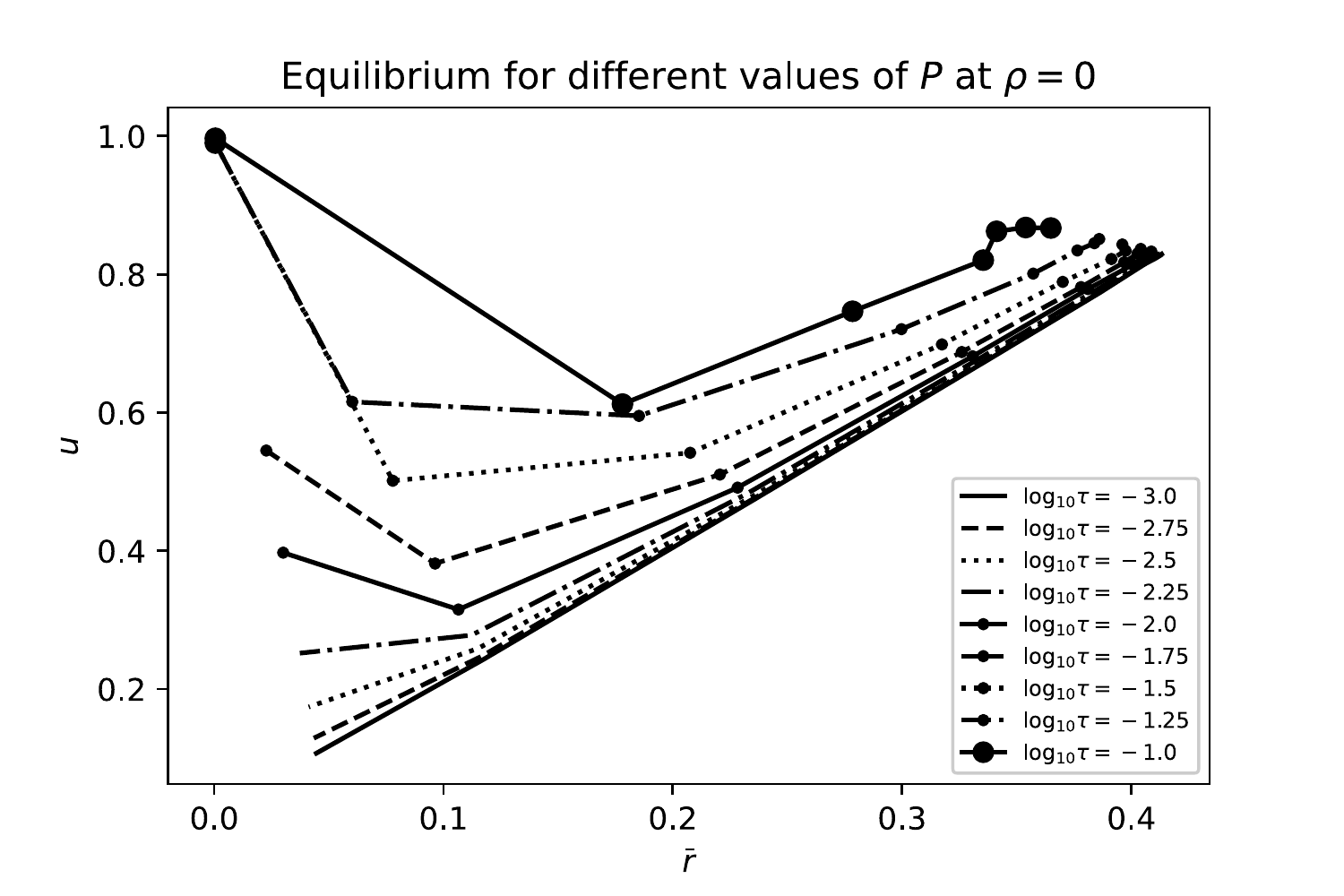}
  \end{minipage}
  \caption{
  The two figures depict the relationship between the total reward $u$ and the average risk taken $\bar r$ under different penalty $P$ and market friction $\tau$ settings. The left figure shows constant $P$ level lines, while the right figure shows constant $\tau$ level lines. We observe a linear relationship between $u$ and $\bar r$ when changing $\tau$ at constant $P$. We find that this behavior holds for all values of $\rho$. Interestingly, there seems to exist a linear relationship between $\bar r$ and $u$ when changing $\tau$ for fixed levels of $P$.}
  \label{fig:result-noise}  
\end{figure}

 The correlation between firm risks also has a significant impact on risk-taking behavior and performance. Specifically, players take more risks in environments of negative correlation, which can improve their payoff, while they take less risks in environments of positive correlation. The impact of correlation on performance is negative in efficient markets, but it can become positive in noisy markets, as illustrated in Figure \ref{fig:result-corr}. These findings provide valuable insights into the complex interplay between market structure, risk-taking behavior, and performance in a competitive environment.

\begin{figure}[htbp]
  \centering
  \begin{minipage}[t]{0.48\textwidth}
    \centering
    \includesvg[width=1.1\linewidth]{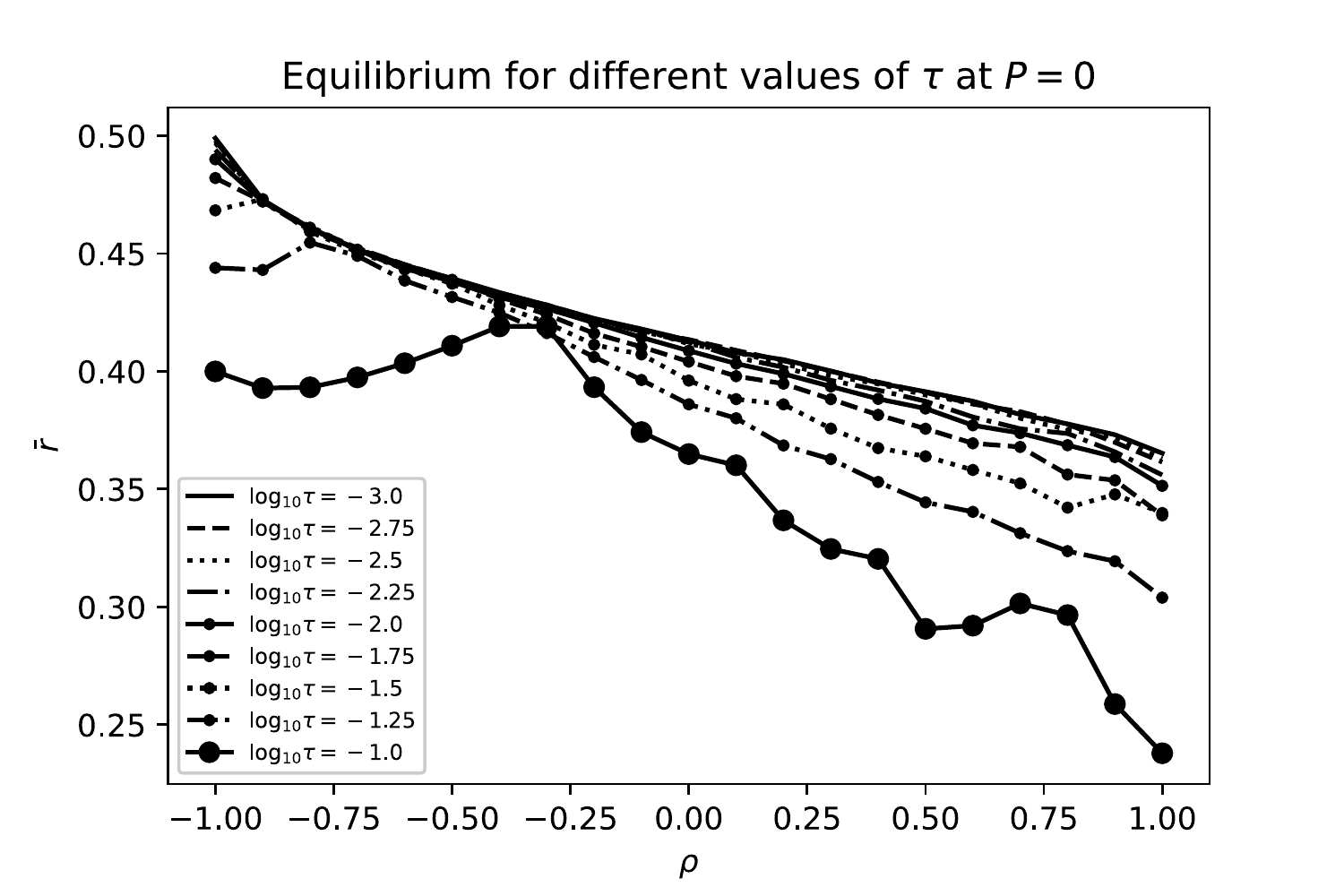}
  \end{minipage}
  \hfill
  \begin{minipage}[t]{0.48\textwidth}
    \centering
    \includesvg[width=1.1\linewidth]{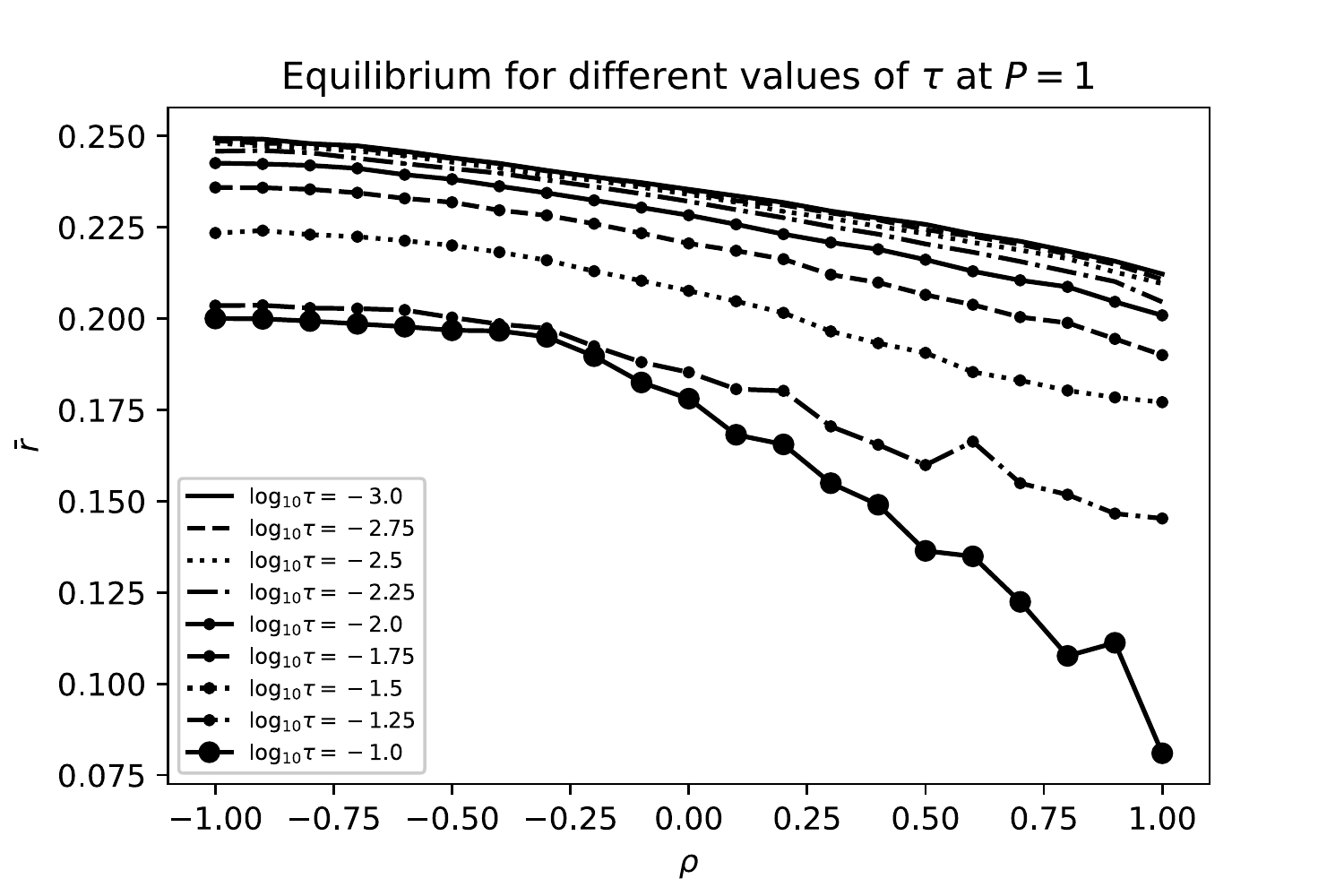}
  \end{minipage}
  \begin{minipage}[t]{0.48\textwidth}
    \centering
    \includesvg[width=1.1\linewidth]{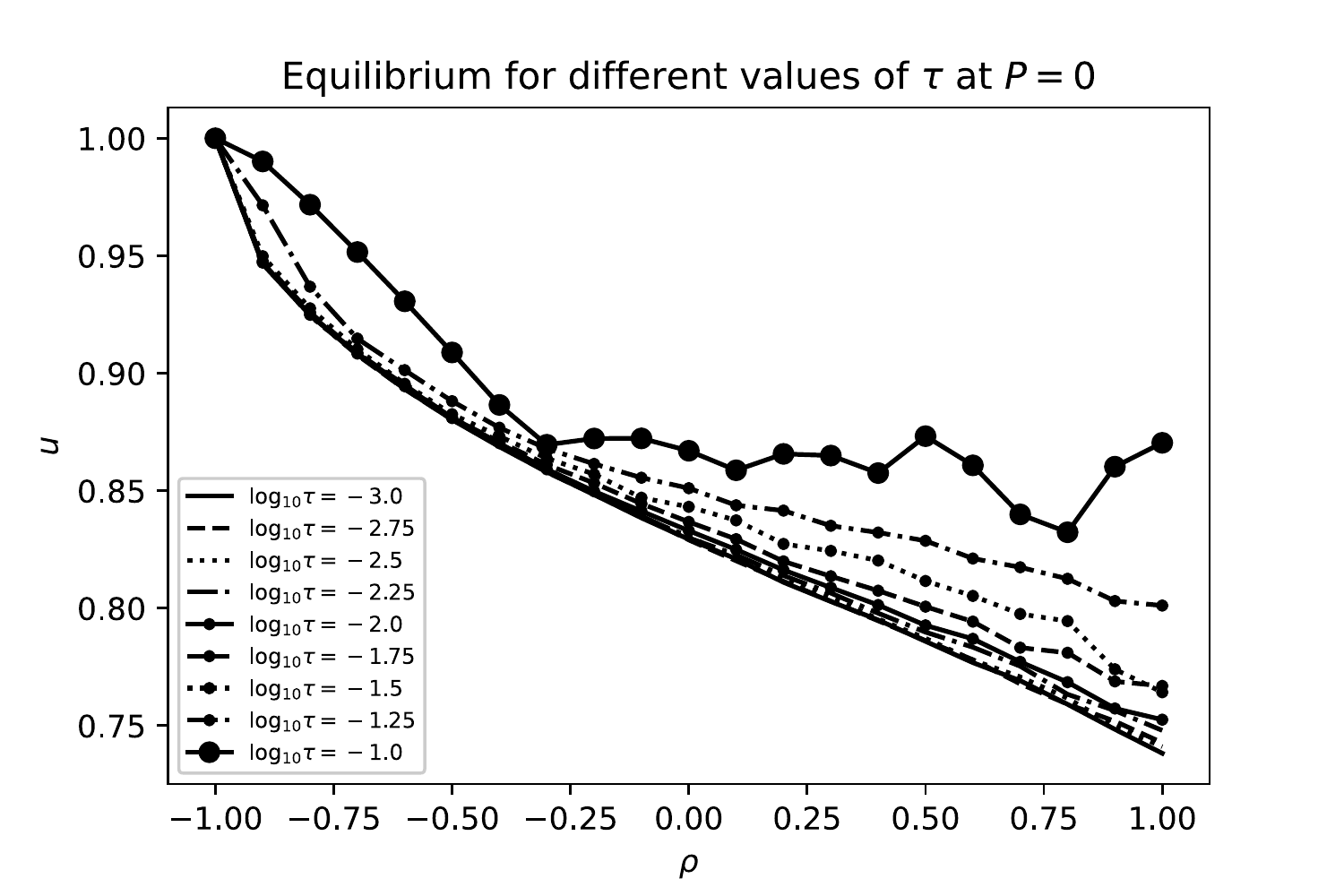}
  \end{minipage}
  \hfill
  \begin{minipage}[t]{0.48\textwidth}
    \centering
    \includesvg[width=1.1\linewidth]{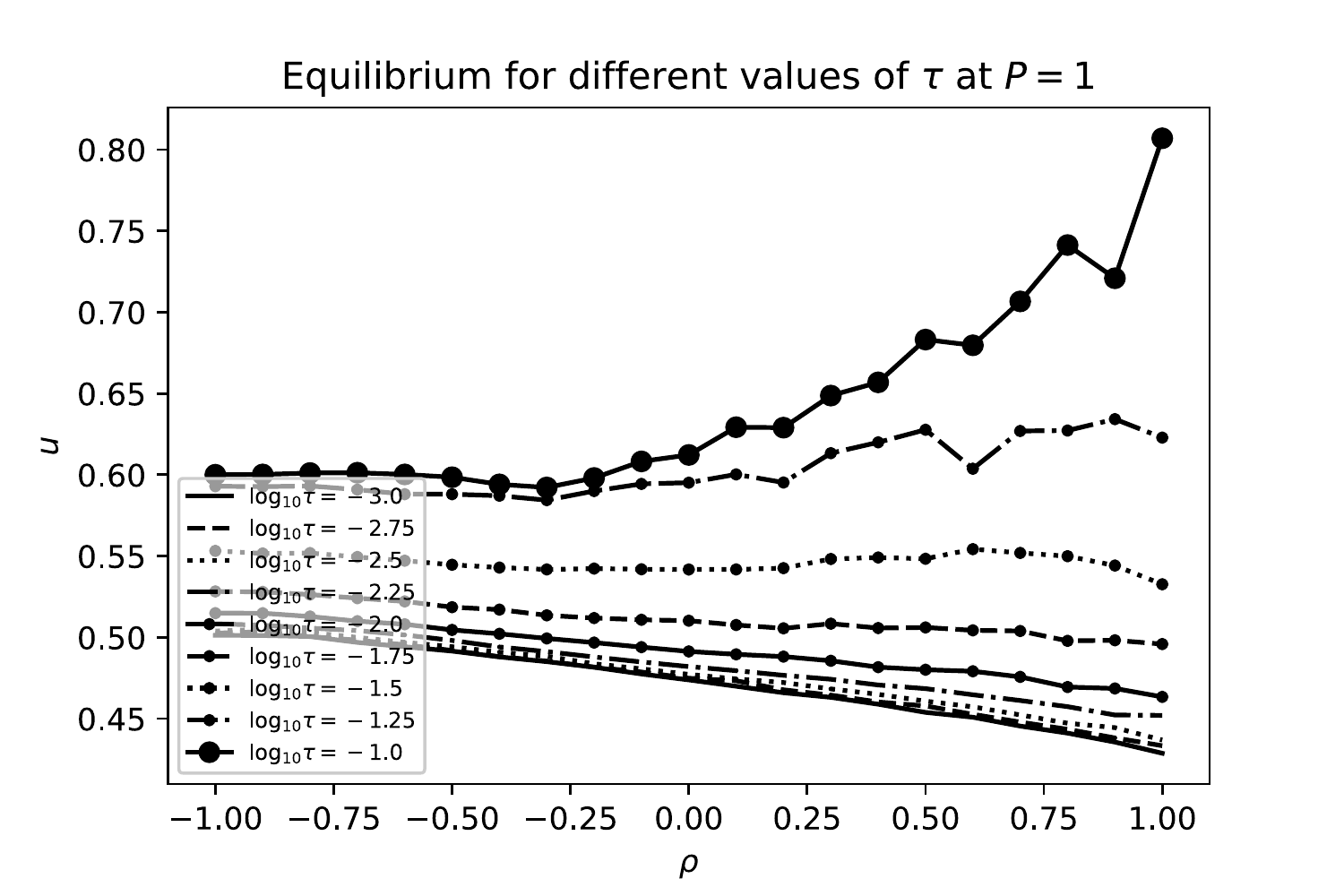}
  \end{minipage}
  \caption{We plot the average risk taken $\bar r$ and the total utility $u$ as functions of the correlation $\rho$ between firm risks.}
  \label{fig:result-corr}  
\end{figure}

\section{Economic Policy Implications and Real-World Case Studies}

In this section, we aim to connect the insights gleaned from our model with broader economic literature and real-world policy implications. Drawing on established economic theories, particularly the Diamond-Dybvig model, we offer a new perspective on how financial institutions manage risk and competition. Further, we delve into how our model's findings can inform current policy debates, such as the regulation of financial advisory commissions in the European Union and contract transferability among life insurers in France. Through this discussion, we illustrate the practical applicability and potential impact of our theoretical framework on shaping economic policies and practices.

\subsection{Ties with the Diamond-Dybvig model}

The Diamond-Dybvig model \citep{diamond1983bank} plays a central role in our understanding of financial intermediaries, especially banks. This seminal model represents a situation where banks provide liquidity transformation services. In a nutshell, it portrays a setting in which depositors need to decide whether to withdraw their money early (impatient depositors) or keep their money in the bank until a later period (patient depositors). Banks offer a contract that allows depositors to withdraw their money early, but at a lower return than if they wait until the end. This can lead to a bank run if too many depositors decide to withdraw their money early.

By allowing financial institutions in our model to ``choose'' a parameter that influences both a utility function and their likelihood of bankruptcy, our approach shares similarities with the Diamond-Dybvig framework. This feature makes our model particularly useful for studying phenomena such as competition among banks over deposit contract interest rates, a scenario that mirrors the dynamic modeled by Diamond and Dybvig.

Despite its parallels with the Diamond-Dybvig framework, our model is distinct, original, and more expansive in scope. It is unconcerned with specific financial metrics or measures, making it more generalizable across different settings. Instead, our model pivots around the broadly applicable concept of ``failure probability''. This fundamental characteristic allows us to abstract away from the complexities of real-world financial instruments and focus on the core strategic interactions of players. By viewing the competition through the lens of failure probability, we can derive insights that are not confined to specific financial instruments or markets, but instead provide a versatile theoretical tool that can be applied across various sectors and scenarios. This innovative feature enhances the relevance and applicability of our model in analyzing strategic risk-taking behavior.

Under certain circumstances, as we show in our study, market frictions, such as customers not optimizing their choice of bank, can increase bank profitability. Although this finding may seem intuitive, our formal model provides a rigorous foundation for this result, focusing not on explicit interest rates but rather on failure probabilities.

Interestingly, our results also shed light on the implications of risk correlation for the equilibrium. In cases where the failure events are negatively correlated, market frictions appear to have little effect on the equilibrium. In contrast, when risks are positively correlated – mirroring the reality of endogenously generated financial risks tied to financial markets – a lack of market friction can spur banks to take on greater risk. This suggests that the optimal situation for both customers and banks arises when the correlation between endogenous risks is zero or negative. Under these conditions, lower frictions can benefit customers without harming banks. On the other hand, if risks are positively correlated, market frictions can harm customers while benefiting banks, creating an incentive for banks to increase these frictions (e.g., by fostering customer loyalty or withholding information).

Moreover, the impact of failure penalties (or the absence of a safety net) can exacerbate this dynamic, suggesting that regulators could mitigate perverse incentives by reducing the penalty parameter $P$ in a positively correlated environment, such as through bailouts or other protective measures (e.g., deposit insurance, liquidity provision, and bank resolution mechanisms). This perspective aligns with a strand of literature that extrapolates from the Diamond-Dybvig model to inform policy making \citep{bhattacharya1985preference,ennis2009bank}.

\subsection{Practical applications in current policy debates}

Outside of the banking sector, our framework has implications for other policy debates. In the European Union, for instance, the Financial Services Commissioner, Mairead McGuinness, recently proposed banning ``inducements'' – commissions paid by banks or insurers to financial advisors who sell their products \citep{Jones2023}. Proponents argue this would enhance transparency and reduce costs. Critics, however, fear it could inhibit access to financial advice. Within our model, increasing transparency equates to reducing the friction parameter $\tau$, which alters the equilibrium, potentially heightening risk for financial institutions but also augmenting returns for consumers. Our findings indicate that regulators might safeguard consumers and influence prices simply by mandating transparency. However, we caution that the impact of such a policy on market frictions is uncertain, as it could inadvertently limit access to information via advisors.

Furthermore, in France, a legislative proposal by MPs Husson and Montgolfier \citep{senat2023} seeks to enhance the transferability of contracts between life insurers, effectively reducing market frictions. Our framework suggests that, in response, insurers could take on more risk due to increased competition. However, the relationship between risk and interest rates may also change, as performance becomes costlier in terms of risk when the duration of contracts decreases because of enhanced transferability. Consequently, while our model predicts an uptick in the risks assumed by insurers, the ultimate benefit to savers remains ambiguous.

\section{Conclusion}

In this study, we thoroughly examined competition models where participants strategically choose their risk levels, with those who take on more risk potentially outperforming their competitors. 

We devised and tested multiple algorithms to solve our game in its discrete form, with vanilla Regret Matching proving to be the most effective. We utilized this efficient implementation to delve into the impacts of penalties, market frictions, and risk correlations on strategic behavior and overall performance. Additionally, we scrutinized the effectiveness of diverse interventions and policies within this competitive landscape.

Our research revealed that market frictions tend to lower the average risk taken while boosting the total reward. Moreover, we found that enhancing failure penalties can foster cooperation and augment the total reward, particularly in inefficient markets. Our exploration also showed that negative correlations among failure events stimulate risk-taking, while positive correlations may discourage it in efficient markets but potentially encourage it in less predictable, noisy markets.

One noteworthy aspect of our study is its parallel with the Diamond-Dybvig framework. Our model, similar to Diamond-Dybvig, examines how financial institutions select parameters influencing their utility functions and likelihood of failure. However, our model is more generalized, focusing on universally applicable notions of failure probabilities, thereby enabling us to study the dynamics of strategic competition in a broader array of scenarios.

We also showcased the adaptability of our model for policy exploration. By imposing policy measures such as transparency requirements or facilitating contract transferability, we demonstrated how policy changes can influence the equilibrium of risk-taking and consequently, the rewards.

Our findings offer substantial insights for economics, finance, and policymaking. By understanding how market frictions and penalties influence competition, firms and governments can make more informed strategic decisions leading to more efficient markets. Moreover, our use of algorithmic solvers for games with continuous action sets illustrates the potential for handling more intricate models lacking closed-form solutions. 

In conclusion, our work provides a robust framework for modeling and analyzing strategic interactions in continuous action games, extending its implications far beyond to enrich economic research and practice.

\section*{Acknowledgement}

We thank \identifier{Aurélie Coursimault}, \identifier{Professor Hélyette Geman}, \identifier{Marc Lanctot}, \identifier{Jules Pondard} and \identifier{Professor Philippe Raimbourg} for helpful discussions and useful comments.

\FloatBarrier
\bibliographystyle{elsarticle-harv} 
\bibliography{refs}

\appendix

\section{Proof of Theorem \ref{thm:nashcor}}
\label{proof:nashcor}

We follow broadly the same scheme as Theorem 3.6 of \citet{Lotker2008-tx}, mutatis mutandis since the equations are different. Their proof proceeds by supposing the existence of a Nash equilibrium and deriving properties to characterize it. Contrary to what they claim, we can prove the existence of a Nash equilibrium without any computation using \citet{Dasgupta1986-gu}. Let $(f_1, f_2)$ be the density functions of Player 1 and 2 in a Nash equilibrium. We note $S_1, S_2$ their support.
\begin{proposition}
    
For almost all $x \in S_1$, $f_2(x) \sim \frac{1}{(1-x)^3}$ and conversely.
\end{proposition} 
\begin{proof}
Let us simplify the notations by noting $a \sim f_1$ and $b \sim f_2$ the random moves of players 1 and 2.  We note $u_1^*$ the utility of Player 1 and $\bar b$ the expectation of $b$. Then, by Theorem \ref{thm:jules}, for almost all $a\in S_1$,
\begin{align*}u_1^* = u_1(a) &= \mathbb{E}_b[u_1(a, b)]\\
&= \mathbb{E}_b\left[Rb(1-a) - aP  + [a>b](1-a)(1-b)R\right]\\
&= R \bar b (1-a) - a P + (1-a) \int_0^a (1-b) R f_2(b) db\end{align*}
$$\int_0^a (1-b) R f_2(b) db = \frac{u_1^* + aP}{1-a} + R \bar b$$

We derivate to obtain $f_2(a) = \frac{u_1^* + P}{R(1-a)^3} $.
\end{proof}
\begin{proposition}
    With the exception of a set measure zero, $S_1 = S_2$.
\end{proposition} 
\begin{proof}We apply the previous result that implies $x \in S_1 \implies f_2(x) \neq 0$.\end{proof}
\begin{proposition}
    $\inf S_1 = \inf S_2 = 0$
\end{proposition}
\begin{proof}
Suppose $\inf S_1 = \inf S_2 = l > 0$. By Theorem \ref{thm:jules}, $u_1(l) = u_1^* = \max_x u_1(x)$. But $u_1(0) > u_1(l)$ since playing $0$ decreases the risk of Player 1 without compromising its chances to get the reward. By contradiction, $\inf S_1 = \inf S_2 = 0$.
 \end{proof}
\begin{proposition}
For all intervals $[x_1, x_2]$ with $0 < x_1 < x_2 < \sup S_1$ we have that $\int_{x_1}^{x_2} f_1(x) dx > 0$
 \end{proposition}
\begin{proof}
    
 Suppose that there is an interval $[x_1, x_2]$ such that $\int_{x_1}^{x_2} f_1(x) dx =0$. Assume that this interval is maximal so that $x_1, x_2 \in S_1$. We also have that $\int_{x_1}^{x_2} f_2(x) dx = 0$ since $S_1 = S_2$. Hence Player 2 never plays between $x_1$ and $x_2$ and this implies that $u_1(x_1) > u_1(x_2)$ since Player 1 can decrease their risk without compromising its chances to get the reward. This contradicts the fact that $u_1(x_1) = u_1(x_2)$ since they are both in the support of the Nash equilibrium.
\end{proof}
\begin{proposition}
There is no point $x$ with positive probability.
\end{proposition}
\begin{proof}
 Suppose the existence of a point $x$ with positive probability (which means that $f_1(x)$ is a Dirac. Since we determined the expression of $f_1(x)$ for almost every $x$, there is a $\varepsilon > 0$ such that there is no other point with positive probability in $[x, x+\varepsilon]$. Hence, there is $0 < \varepsilon' < \varepsilon$ such that $u_2(x + \varepsilon') > u_2(x)$ since any positive $\varepsilon'$ improves the probability of winning the reward by $\mathbb{P}[\text{Player 1 plays } x]$ and the risk increment goes to $0$ with $\varepsilon'$.\
\end{proof}

\nashcor*

\begin{proof}
Up to a set of measure 0 and on which the probability is 0, we have determined the expression of $f_1$ and $f_2$ above up to a constant. This means that $f_1 = f_2 = f$.
We still have to find the constant and the upper limit of the support.
Let us reiterate the calculations, this time using the fact that $0 \in S$ with $S$ the support of $f$. This means that for any $a\in S$:

\begin{align*}0 &= u_1(a) - u_1(0)\\
&= \mathbb{E}_b[u_1(a, b) - u_1(0, b)]\\
&= \mathbb{E}_b\left[Rb(1-a) - aP  + [a>b](1-a)(1-b)R - Rb \right]\\
&= -a R \bar b - a P + (1-a) \int_0^a R (1-b) f(b) db\end{align*}

$$\int_0^a R (1-b) f(b) db = (R \bar b + P) \frac{a}{1-a}$$

We derivate to obtain:
$$f(a) = \left(\bar b + \frac{P}{R}\right) \frac{1}{(1-a)^3}$$
We now have two unknowns and two unknowns:
$$
\left\{ \begin{array}{ll}
h := \sup S  \text{ such that } \int_0^h f(x) dx = 1\\
\bar b = \int_0^h x f(x) dx
\end{array} \right.
$$
We define $P := \frac{P}{R}$.\\
$$\int_0^h f(x) dx = \frac{\bar b + P}{2} \left(\frac{1}{(h-1)^2}-1\right) = 1$$\\
$$\int_0^h x f(x) dx = \bar b = \frac{\bar b + P}{2} \frac{h^2}{(h-1)^2}$$

We get
$$\bar b = \frac{h^2}{1-(h-1)^2}$$
$$2 \frac{(h-1)^2}{1 - (h-1)^2} - P = \frac{h^2}{1 - (h-1)^2}$$
$$2(h-1)^2 - P + P(h-1)^2 - h^2 = 0$$
$$(P+ 1) h^2 - (2 P+ 4) h + 2 = 0$$
Hence, $h = \frac{2 + P\pm \sqrt{P^2 + 2 P+ 2}}{1 + P}$\\
Since $h<1$, we get
\begin{align*}
    h &= \frac{2 + P- \sqrt{P^2 + 2 P+ 2}}{1 + P}\\
      &= 1 - \frac{\sqrt{(P+ 1)^2 + 1} - 1}{1+P}
\end{align*}
We remark that $(\sqrt{(P+ 1)^2 + 1} + 1)( \sqrt{(P+ 1)^2 + 1} - 1) = (P+1)^2$\\
Hence $h = 1 - \sqrt{\frac{\sqrt{(P+ 1)^2 + 1} - 1}{\sqrt{(P+ 1)^2 + 1} + 1}}$
\begin{align*}\bar b + P&=\frac{2}{\frac{1}{(h-1)^2} - 1}\\
&= \frac{2}{\frac{\sqrt{(P+ 1)^2 + 1} + 1}{\sqrt{(P+ 1)^2 + 1} - 1} - 1}\\
&= \sqrt{(P+ 1)^2 + 1} - 1\end{align*}
Finally, $f(x) = \left[x < 1 - \sqrt{\frac{k - 1}{k + 1}}\right] \frac{ k - 1}{(1-x)^3}$ with $k := \sqrt{(P+ 1)^2 + 1}$\\
From the expressions above, we get $\bar b = k - 1 - P$ and $u^* = R \bar b$.
\end{proof}

\section{\texorpdfstring{Estimation of ${r_{max}}$ and $w$ in the multiplayer setting}{Estimation of rmax and w in the multiplayer setting}}
\label{estimate}

We recall the equations:

$$f(x) = \frac{P+ w}{(n-1)(1-x)^{2+\frac{1}{n-1}} (Px + w)^{1 - \frac{1}{n-1}}}$$
\begin{align*}
        \int_0^{r_{max}} f(x) dx &= 1 \\
        \int_0^{r_{max}} x f(x) dx &= \bar r 
\end{align*}
with $$w := \bar r ^ {n-1}$$

We could estimate the integrals by numerical integration. However $\bar r$ becomes smaller when increasing $n$, which causes $f$ to have a peak at $0$ and renders the integral estimates unreliable.
Computations \footnote{We used Wolfram Alpha with \texttt{ReplaceAll[Integrate[(1/(1 - x)) D[((w + x P)/(1 - x))\^{}(1/(n - 1)), x], x], {n -> 42}]} for various values of $n$, found a pattern and checked that the derivatives match.} give us 
$$\int f(x) = \frac{w + nP(1-x) + Cx}{n (1 - x)(P+w)} \sqrt[n-1]{\frac{Cx + w}{1-x}}$$
$$\int x f(x) = \frac{w - n w  (1-x) + Cx}{n (1-x)(P+w)}\sqrt[n-1]{\frac{Cx+w}{1-x}}$$
Hence
\begin{equation}
\label{eq:multi1}
\frac{w + nP(1-{r_{max}}) + P{r_{max}}}{n (1 - {r_{max}})(P+w)} \sqrt[n-1]{\frac{P{r_{max}} + w}{1-{r_{max}}}} -  \frac{w + nP}{n (P+w)} \sqrt[n-1]{w} = 1
\end{equation}
\begin{equation}
\label{eq:multi2}
\frac{w - n w  (1-{r_{max}}) + P{r_{max}}}{n (1-{r_{max}})(P+w)}\sqrt[n-1]{\frac{P{r_{max}}+w}{1-{r_{max}}}} - \frac{w - n w}{n (P+w)}\sqrt[n-1]{w} = \sqrt[n-1]{w}
\end{equation}

The value of integrals are increasing in ${r_{max}}$ since $f(x)$ is positive, hence for any value of $w$ we can find the corresponding value of ${r_{max}}$ by binary search using \eqref{eq:multi1}. Then we are left with finding the value of $w$ using \eqref{eq:multi2} . We observe experimentally that the function $w \rightarrow \int_0^{{r_{max}}(w)}xf(x)dx - \sqrt[n-1]{w}$ seems to have only one root and that it is positive before that root and negative after that root. We therefore use binary search.

Equation \eqref{eq:multi2} also defines $w$ as a fixed point and the iterative algorithm $w \leftarrow \left(\int_0^{{r_{max}}(w)}x f(x)\right)^{n-1}$ also converges, although it seems to sometimes loop between a few close values because of numerical errors.

We remark that $w$ goes very quickly to $0$ as $n$ increases. This causes numeric errors in the computations of the primitives at $0$ that we fix by storing $\log w$ instead of $w$.  For the same reason, we compute the integrals in log space.

Equation \eqref{eq:multi2} is more interesting as the indefinite integral $F({r_{max}})$ can be negative. Depending on the sign, we store either $\log F({r_{max}})$ or $\log -F({r_{max}})$.

Finally, we apply the Log-Sum-Exp Trick to compute the value of the integral in log space. The final algorithm to find $w$ and ${r_{max}}$ takes less than 1 ms on a laptop for any value of the parameters.

\section{Proof of Theorem \ref{thm:nashwithcor}}
\label{proof:nashwithcor}

\nashwithcor*
\begin{proof}
        We reuse the proof of \ref{proof:nashcor}, only modifying the calculations.
    We note $a$ and $b$ the actions of the players, aka their individual probability of failure noted $r_1$ and $r_2$ above and note $c$ their joint probability of failure noted $\tilde r$ above. $c$ is a function $c(a,b,\rho)$. We note $p$ the probability distribution corresponding to the Nash equilibrium.
    In both cases, we have the equality $$\int_0^1 p(b) c~db + a P= \int_0^a p(b)(1-a-b+c)~db$$

    When $\rho = 1$, $c = \min(a, b)$

    $$\int_0^a p(b) b~db + a \int_a^1 p(b)~db + a P= \int_0^a p(b)(1-a)~db$$

    We derivate wrt $a$:

    $$a p(a) + \int_a^1 p(b)~db - ap(a) + P= p(a)(1-a) - \int_0^a p(b)~db$$

    $$p(x) = \frac{1+P}{1-x}$$

    We solve $\int_0^{r_{max}} p(x)~dx=1$ as $\log(1-r_{max}) = -\frac{1}{P+1}$, thus
    $$r_{max} = 1 - \exp\left(-\frac{1}{P+1}\right)$$
    Finally, a simple computation gives $$\bar r = 1 - (P+1)\left(1-\exp\left(-\frac{1}{P+1}\right)\right)$$

    When $\rho = -1$, $c = \max(0, a+b-1)$. In other terms, $c=0$ if $b<1-a$ and $c=a+b-1$ if $b>1-a$.

    We first treat $a>\frac{1}{2}$, aka $a>1-a$, to show by contradiction that $p(a) = 0$.
    $$\int_{1-a}^1 p(b)(a+b-1)~db+ aP= \int_0^{1-a} p(b)(1-a-b)~db$$
    $$aP= \int_0^1 p(b)(1-a-b)~db = 1 - \bar r - a$$
    This is impossible, thus $p(a) = 0$ for $a> 1/2$.

    We now suppose $a < 1-a$:

    $$\int_{1-a}^1 p(b)(a+b-1)~db + aP= \int_0^a p(b) (1-a-b)~db$$

    $\int_{1-a}^1 p(b)(a+b-1)~db = 0$ since $1-a >  \frac12$.
    We derivate twice to get:

    $$(1-2a)p(a) = \int_0^a p(b)~db + P$$
    $$1-2a) p'(a) - 3p(a) = 0$$

    The first equation gives $p(0) = P$. Combined with the second differential equation, we get:

    $$p(x) =\frac{P}{(1-2x)^{3/2}}$$

    Similarly, we solve  $$r_{max} = \frac{1}{2} - \frac{P^2}{2(P+1)^2}$$
    $$\bar r = \frac{1}{2P+2}$$
\end{proof}
\section{Proof of Theorem \ref{thm:pvalue}}
\pvalue*
\begin{proof}
    We suppose that the diameter $d$ is obtained in some unit direction $\vec a$: $$\max \limits_{x\in P}a \cdot x - \min \limits_{x\in P}a \cdot x = d$$
    We define $D$ the line $\left[\argmin \limits_{x\in P}a \cdot x, \argmax \limits_{x\in P}a \cdot x\right]$.
    Then, for any $v$, $$\max \limits_{x\in P}v \cdot x - \min \limits_{x\in P}v \cdot x \ge \max \limits_{x\in D}v \cdot x - \min \limits_{x\in D}v \cdot x = d \times |v \cdot a|$$

Hence, $$\Pr\left[ \max \limits_{x\in P}v \cdot x - \min \limits_{x\in P}v \cdot x < \varepsilon \right] \le \Pr\left[ |v \cdot a| < \frac{\varepsilon}{d} \right]$$

If $v$ is a vector of iid standard Gaussian variables, $v \cdot a \sim \mathcal{N}(0, 1)$ and $(v \cdot a)^2$ follows a $\chi^2$ distribution with 1 degree of freedom.
The sum of multiple iterations will follow:

$$\Pr\left[\sum_i^K (\max \limits_{x\in P}v_i \cdot x - \min \limits_{x\in P}v_i \cdot x)^2 < \varepsilon \right] \le \Pr\left[ \sum_i^K z_i^2 < \frac{\varepsilon}{d} \right]$$

where $z_i \sim \mathcal{N}(0, 1)$ are independent normal variables.

$\sum_i^K z_i^2$ follows a $\chi^2$ distribution with $K$ degrees of freedom. With $F_{\chi^2}(x, K)$ the cumulative distribution function, we get:

$$\Pr\left[\sum_i^K (\max \limits_{x\in P}v_i \cdot x - \min \limits_{x\in P}v_i \cdot x)^2 < \varepsilon \right] \le F_{\chi^2}\left(\frac{\varepsilon}{d}, K\right)$$
\end{proof}

\section{Linear equilibrium in finite approximations with \texttt{shift~=~false}}
\label{sec:linear-noshift}

\begin{figure}[H]
  \centering
  \begin{minipage}[t]{0.48\textwidth}
    \centering
    \includesvg[width=1.1\linewidth]{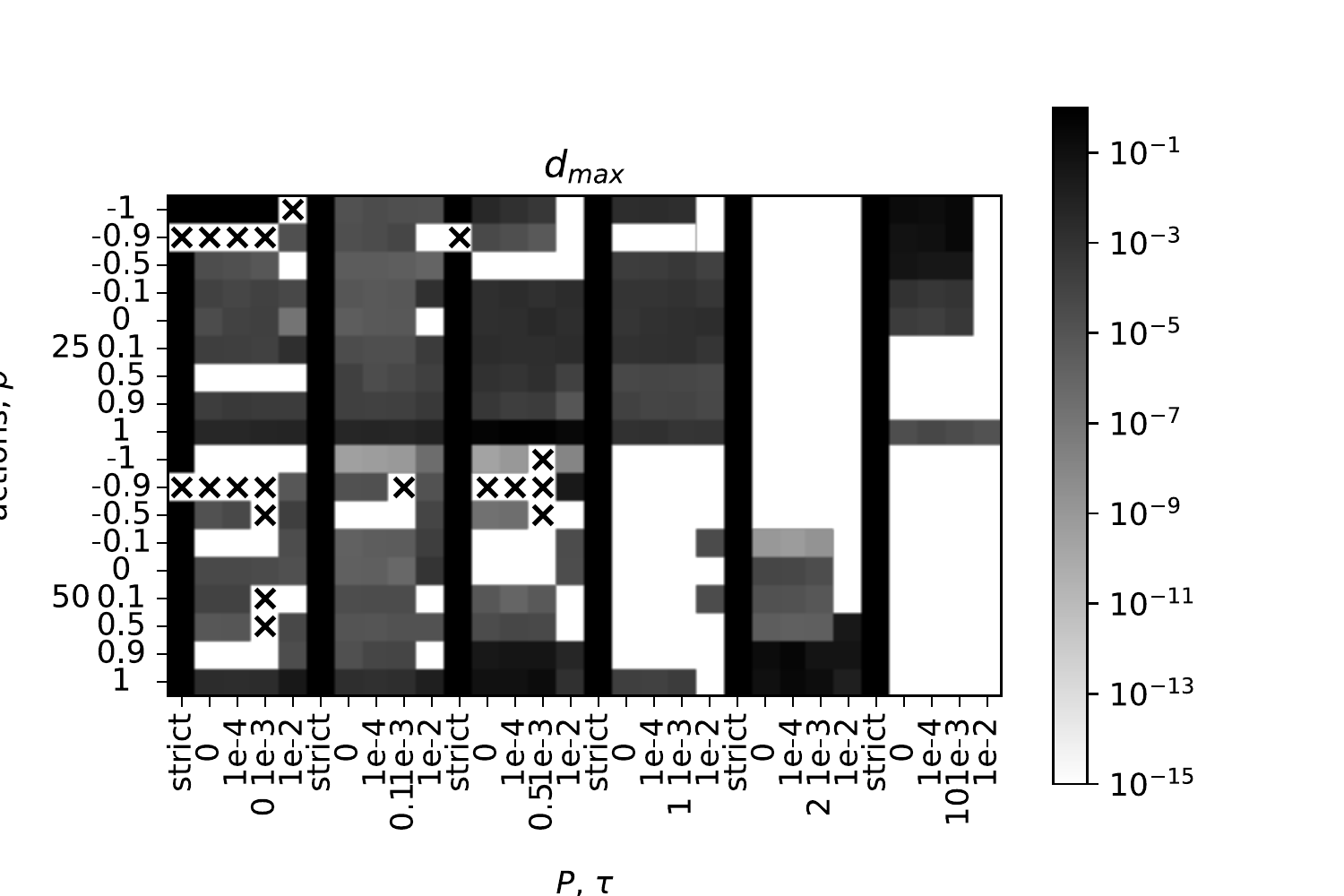}
  \end{minipage}
  \hfill
  \begin{minipage}[t]{0.48\textwidth}
    \centering
    \includesvg[width=1.1\linewidth]{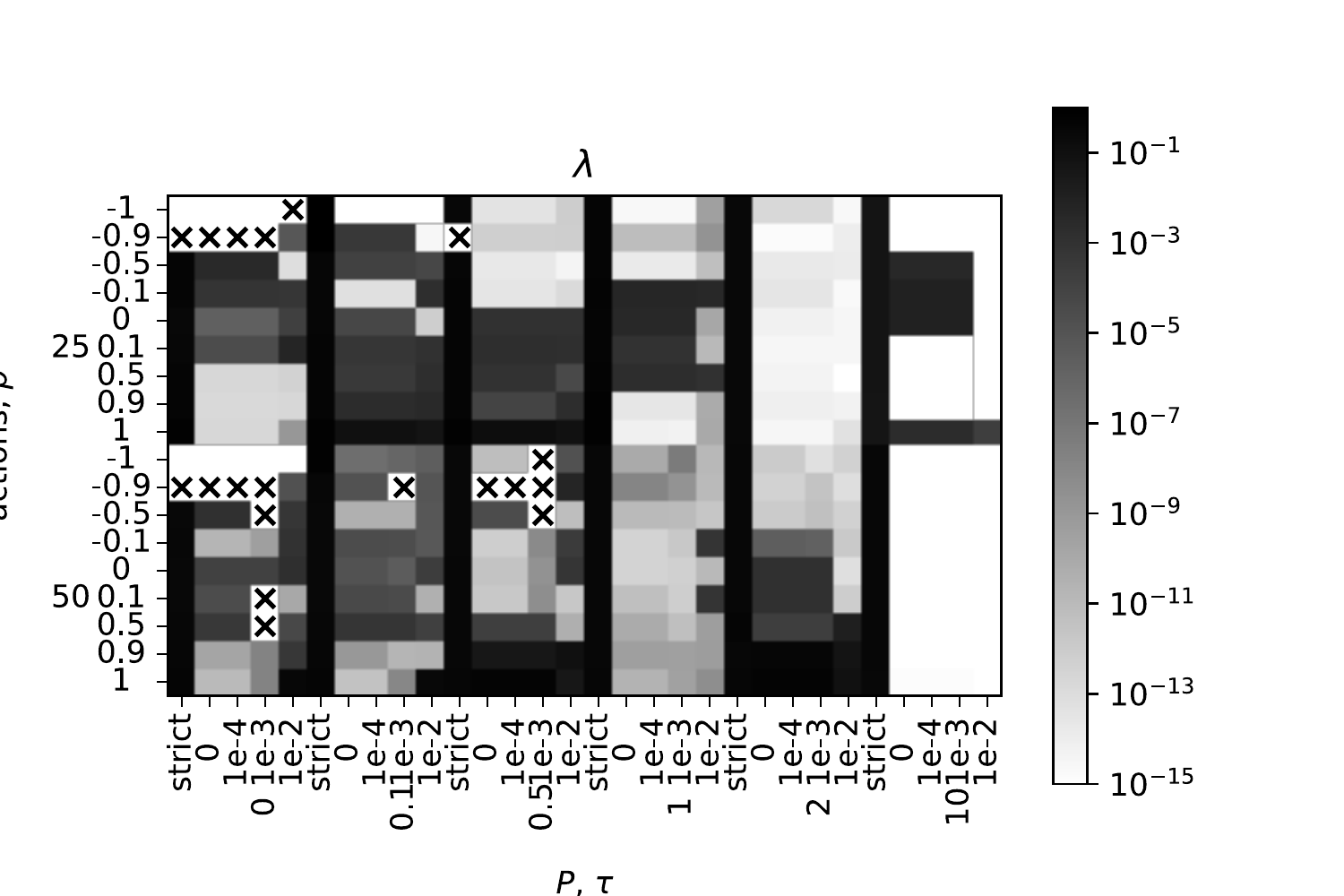}
  \end{minipage}
\end{figure}





\end{document}